\documentclass[final,5p,times]{elsarticle}

\usepackage{amssymb}
\usepackage{amsmath}
\usepackage{microtype}
\usepackage{natbib}
\usepackage{threeparttable}
\usepackage{hyperref}
\usepackage{algorithm}
\usepackage{algpseudocode}
\usepackage{booktabs}
\usepackage{bbold}
\usepackage{xcolor}
\usepackage{amsthm}
\newtheorem{theorem}{Theorem}
\newtheorem{assumption}{Assumption}
\newtheorem{lemma}{Lemma}

\newcommand{\jr}[1]{{\color{black} #1}}

\begin{document}

\begin{frontmatter}



\title{Load Management of Distribution Systems via Online Dynamic Pricing}



\author[inst1]{Jiarui Yu}

\affiliation[inst1]{organization={Laboratoire d'Automatique, EPFL},
            postcode={1015}, 
            state={Lausanne},
            country={Switzerland}}

\author[inst2]{Zhiyu He \corref{cor1}}
\author[inst1]{Wenbin Wang}
\author[inst1]{Colin N. Jones}
\author[inst2]{Florian Dörfler}
\author[inst3]{Hanmin Cai}

\cortext[cor1]{Corresponding author at: Automatic Control Laboratory, ETH Zürich, Switzerland.
 E-mail address: \nolinkurl{zhiyhe@control.ee.ethz.ch}.}

\affiliation[inst2]{organization={Automatic Control Laboratory, ETH Zürich},
            postcode={8092}, 
            state={Zürich},
            country={Switzerland}}

\affiliation[inst3]{organization={Urban Energy Systems Lab, EMPA},
            postcode={8600}, 
            state={Dübendorf},
            country={Switzerland}}

\begin{abstract}
The growing adoption of electric vehicles (EVs) is increasing peak demand in distribution systems, which can threaten grid stability and reduce operational efficiency. Dynamic electricity pricing is a promising means of mitigating these peaks by shifting flexible demand. However, most existing approaches rely on detailed user-level consumption data and behavioral models, which are often difficult to obtain in practice and may raise privacy concerns.

This paper proposes an Online Feedback Optimization (OFO) algorithm for day-ahead price design with limited data, where only aggregate loads are observed. OFO updates prices iteratively using aggregate load measurements, enabling effective peak reduction without access to individual user data. The formulation also includes a term that penalizes deviations in total electricity cost relative to a reference tariff. Although relying only on aggregate load measurements, the OFO price updates efficiently converge to the optimal price. 
In finite-horizon simulations, OFO achieves peak reduction close to that of the Stackelberg benchmark with full model information. Meanwhile, its computational effort is substantially lower. Additional tests under multiple initial conditions and delayed charging-window mismatch further confirm the robustness of the proposed method. Overall, these results show that OFO is a scalable and computationally efficient approach for peak-demand management in distribution systems with limited observability.

\end{abstract}




\begin{keyword}
Online feedback optimization \sep electrical vehicle charging \sep dynamic pricing design
\end{keyword}

\end{frontmatter}


\section{Introduction}
\label{sec:intro}

Modern energy systems are evolving with the increasing integration of flexible loads, larger industrial and commercial demands, and distributed energy resources (DERs) into the distribution network~\cite{LEE_integrating_2026}. 
Within this transition, the rapid growth of flexible electrical loads, most notably electric vehicles (EVs), is reshaping power system operation and introducing new operational challenges.
When left uncontrolled, EV charging can significantly amplify peak demand. Previous studies report peak increases of up to 74.4\% at full penetration~\cite{Jaru_predicting_2023}, which can lead to voltage instability, transformer overloading, and thermal stress on network components~\cite{GAGGERO2025100395}.
\jr{To mitigate such congestion, distribution system operators can employ a range of congestion management approaches, such as static tariffs, dynamic tariffs, scheduled load management, local flexibility markets, and direct load control~\cite{mood_review_2026}. While scheduled load management can coordinate flexible demand explicitly, it requires prior knowledge of device availability, flexibility constraints, and user preferences which may be difficult to obtain for large and heterogeneous consumer populations~\cite{nikos_optimal_2015}. 
Similarly, direct load control, local flexibility markets, and aggregate-control approaches can provide reliable congestion relief by coordinating flexible demand more explicitly, as shown by controllable-load frameworks for balancing and reserve services~\cite{callaway_achieving_2011,ma2013decentralized} and aggregate setpoint-based control of thermostatically controlled loads~\cite{kundu_modeling_2011}. However, such approaches typically require intrusive or semi-direct control, explicit load models, communication infrastructure, or complex market designs, and may raise concerns related to user acceptance, strategic behavior by market participants, or cost recovery~\cite{hennig_congestion_2023}.}
Pricing-based schemes, by contrast, offer a scalable and non-intrusive means of influencing consumer behavior. By leveraging price signals to activate demand-side flexibility, energy retailers and DSOs can reduce congestion without direct intervention or detailed modeling of individual loads~\cite{schreiber_flexible_2015}.
Well-designed pricing strategies have been shown to help smooth demand fluctuations, improve grid efficiency, and reduce costs for consumers~\cite{rasheed_dynamic_2020}.

Despite these benefits, converting pricing-based flexibility into reliable peak reduction at the distribution level remains challenging. In particular, fairness considerations often require uniform price signals to be broadcast to many consumers, which can trigger synchronized responses and produce rebound peaks that undermine system stability~\cite{nazir_load_2017}. Designing pricing mechanisms that balance incentive effectiveness, system stability, and fairness therefore remains an ongoing challenge.

Static schemes such as flat-rate and time-of-day tariffs have been widely used in traditional energy systems but fail to adapt to real-time variations in demand and supply~\cite{faruqui_household_2010}. More advanced methods, such as multi-step electricity price schemes seek to overcome these challenges by considering both time steps and quantity~\cite{lin_multistep_2016}. However, these pricing schemes are not fully dynamic and still depend on fixed peak-hour definitions, which can unintentionally shift demand and create new peaks during supposed off-peak periods. These limitations underscore the need for a more adaptive pricing approach that can align consumer behavior with real-time system conditions.

Existing dynamic pricing research often relies on simplified behavioral models to estimate user responses to price signals. For example, 
utility functions are widely used to characterize user behaviors with different variants like specific user behavior based~\cite{cui_pricing_2025}, logistic based~\cite{li_realtime_2021} and microeconomic principle based~\cite{tarasak_optimal_2011} formulations.
However, utility-based models usually rely on a predefined utility function to describe how consumers respond to price changes. While convenient, this assumption may oversimplify real demand behavior, which often varies across users in both flexibility and consumption needs.
Even with utility function trained with customized historical data~\cite{subramanian_dynamic_2013}, it still struggles when user responses shift unexpectedly due to personal preferences or external factors.
Other simplified models also make strong behavioral assumptions. For instance, stationary users~\cite{lei_realtime_2016}, known properties~\cite{rabie_novel_2026} and perfectly rational decisions~\cite{wang_new_2021} are widely assumed.
Such assumptions limit the robustness of these methods in dynamic environments with uncertainty and incomplete observability.

\begin{figure*}[ht]
    \centering
    \includegraphics[width=\linewidth]{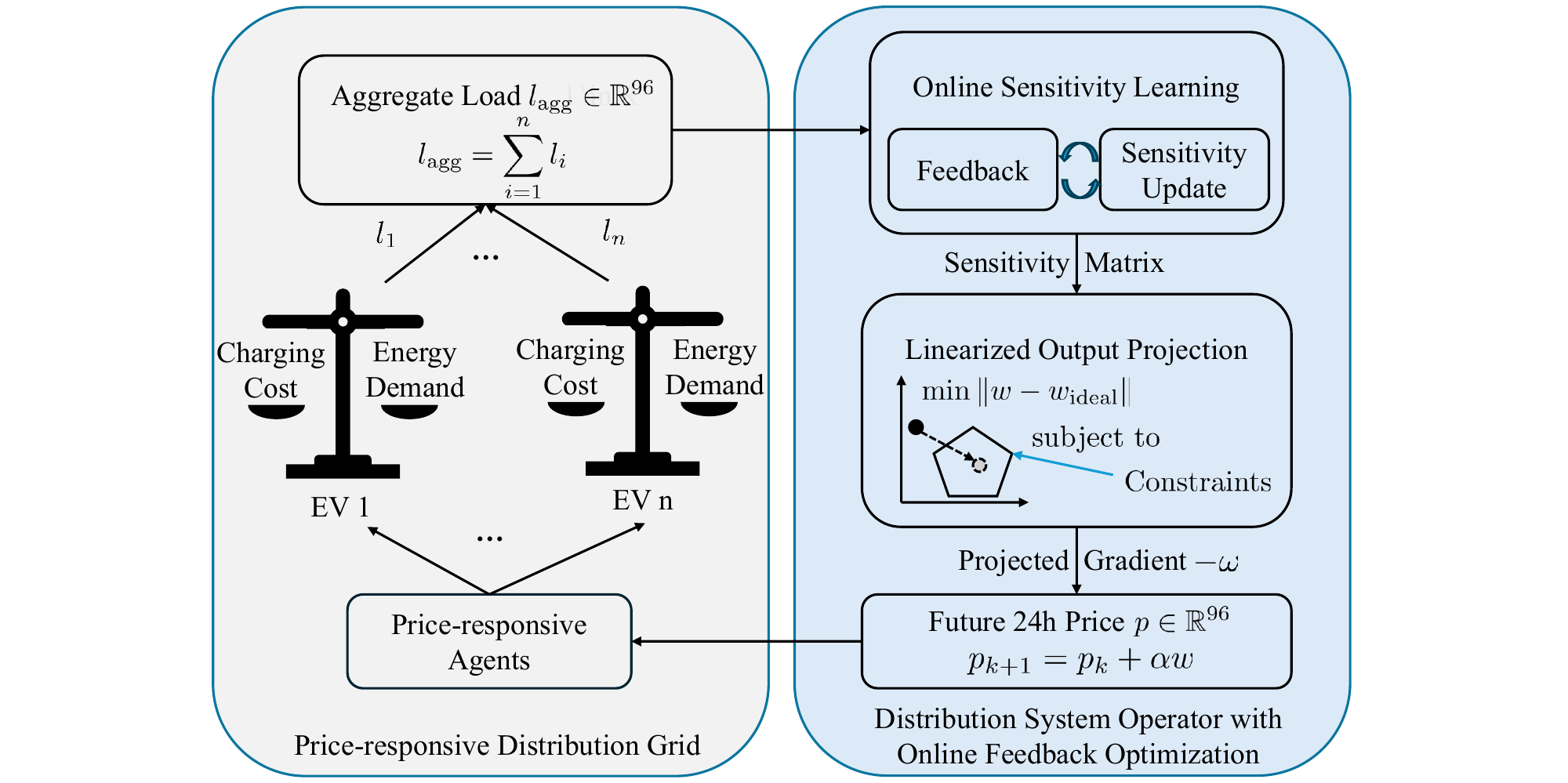}
    \caption{Interaction between the Distribution System Operator and the Distribution Grid via Online Feedback Optimization.}
    \label{fig:prob_statement}
\end{figure*}

To better reflect real consumer behavior, many dynamic pricing approaches adopt data-driven methods that rely on historical or real-time observations. These methods typically estimate user participation and preferences from frequent feedback~\cite{chiu_optimized_2017} or detailed consumption records~\cite{qiu_personalized_2023}, requiring extensive communication and access to individual-level data. While such information can improve pricing accuracy, it raises privacy concerns and increases system complexity.
More recent learning-based approaches~\cite{palaniyappan_dynamic_2024} further leverage high-resolution measurements, such as voltage profiles or power quality indicators, to enhance prediction performance. However, these signals are rarely available at scale in distribution networks and are often subject to privacy, access, or communication constraints. Other strategies may depend on long-term smart-meter data to infer individual responses~\cite{taherian_optimal_2021}, which limits their applicability in environments with partial observability or restricted data access. Overall, although data-driven methods perform well under full observability, their reliance on granular user-level information constrains privacy preservation and hinders scalable deployment in real-world distribution systems.


 \jr{A further concern in dynamic pricing design is implementation complexity, as many existing approaches rely on large-scale or nested optimization procedures that impose computational burdens and may require repeated communication.} Population-based optimization methods, such as particle swarm algorithms, search iteratively over many candidate solutions~\cite{guo_how_2023}, resulting in substantial computational overhead that limits real-time applicability~\cite{poli_particle_2007}. Stackelberg game--based formulations~\cite{erol_stackelberg_2024,jia_dynamic_2016} introduce additional complexity due to their hierarchical leader--follower structure, where the leader's pricing decision anticipates the followers' responses and leads to bi-level optimization problems. These problems must either be solved iteratively using decomposition methods, which require solving two interdependent subproblems or be reformulated as single-level problems with additional integer variables or nonlinear constraints~\cite{behnke_computational_2024}. Both approaches increase computational effort and implementation burden. Recent learning-based extensions further increase computational demands~\cite{hu_mixed_2025}. Although these approaches are conceptually powerful, their high complexity hinders scalability and real-time deployment in realistic distribution grids.

Overall, despite significant progress, many existing pricing methods rely on detailed consumer information, assume simplified or fully rational responses, or require substantial computational resources. In practice, these assumptions and requirements are difficult to satisfy in real distribution systems, where user behavior is heterogeneous, data access is limited, and fast decision making is essential. These limitations underscore the need for a pricing mechanism that can adapt to actual consumer behavior while preserving privacy and remaining computationally scalable under limited observability.

To meet these requirements, this study introduces Online Feedback Optimization (OFO) as a practical approach for day-ahead price design with aggregate data. OFO performs light-weight, gradient based updates using real time feedback from measured system responses instead of relying on detailed behavioral or system models~\cite{hauswirth_optimization_2024}. This feedback-driven and model-free structure enables automatic adaptation to evolving demand, preserves consumer privacy, and significantly reduces computational burden compared with optimization or game-theoretic formulations.
The main contributions of this work are summarized as follows:
\begin{enumerate}
    \item \textbf{Online dynamic pricing algorithm with limited observation:} An OFO approach based on online sensitivity learning is implemented with event-triggered updates and warm initialization to improve convergence speed, robustness, and practical applicability under realistic load dynamics. The method relies only on aggregate load measurements, which preserves user privacy and supports scalable deployment in distribution systems with limited observability.

\item \textbf{Theoretical guarantees of convergence and feasibility:} 
We establish almost sure convergence of the proposed online dynamic pricing scheme to the set of projected stationary points under standard stochastic approximation assumptions. To facilitate clear proof, the convergence analysis is performed on a differentiable smoothened reformulation. The practical implementation and reported case-study results, however, are based on the original formulation, since it preserves the operational objective and constraints of interest. Feasibility with respect to the original constraints is maintained at every iteration.


\item \textbf{Near-optimal performance on high fidelity simulations:} Extensive simulations using daily charging demand from more than 300 EVs show that OFO reduces the last-14-day average daily peak by $25.3\%$ relative to the reference-price baseline. Despite using only aggregate measurements and a limited number of online updates, OFO leaves a remaining peak gap of only $14.6\%$ relative to the offline Stackelberg benchmark with full individual-level model information. \jr{Additional tests under multiple initial conditions and delayed charging-window mismatch further confirm the robustness of OFO.}

\end{enumerate}

\section{Problem Statement}
\label{sec:prob}

\subsection{Overview}

This section defines the pricing problem considered in this work: designing a day-ahead price mechanism that guides flexible demand in distribution grids to reduce peak load while maintaining operational feasibility. The overall structure of the proposed closed-loop system is illustrated in Fig.~\ref{fig:prob_statement}. The Distribution System Operator (DSO) \jr{computes a 96 dimensional price vector $p \in \mathbb{R}^{96}$ corresponding to 24 hour horizon with 15-minute resolution} and broadcasts it to the price-responsive distribution grid. Price-responsive agents, such as EV charging stations and other flexible devices, then individually adjust their demand. Although the proposed framework applies to general price-responsive loads, this study focuses on EV charging as the primary application.
\jr{Their responses $l_i \in \mathbb{R}^{96}$ are summed at the feeder-head or substation-level aggregation point, yielding the total load profile $l_\text{agg} \in \mathbb{R}^{96}$, which is directly available to the DSO through network measurements.}
 This iterative process forms the basis of an Online Feedback Optimization (OFO) loop, enabling price updates based on real-time system responses while using only aggregate information. 
 
 \jr{
 This setting is related to aggregative-game formulations, where agents are coupled through an aggregate quantity and equilibrium-seeking algorithms are designed for the resulting game~\cite{grammatico_dynamic_2017}. Different from this line of work that focuses on static equilibrium problems with known game structures~\cite{belgio_semi_2023}, this work addresses a concrete distribution-grid pricing problem through iterative price updates based on aggregate load feedback, without requiring prior knowledge of individual response models.
 }

\subsection{Price Design Model}
Instead of explicitly solving the problems of agents and DSO in a separate manner, the proposed approach adopts a bi-level structure similar to a Stackelberg game. As shown in problem~\eqref{eq: ofo_formulation}, the DSO sets price signals, and then agents respond by optimizing their individual consumption. Unlike traditional Stackelberg game formulations, the DSO observes only aggregate load feedback rather than individual responses, as shown in Eq.~\eqref{eq: agg}. Therefore, OFO needs to update prices based on partial information and without explicit follower models, which substantially complicates the problem.

\vspace{-1em}
\begin{subequations}
\begin{align}
\text{OFO:} \quad \mathop{\min}_{p\in\mathbb{R}^{96}} \quad & \max(l_{\text{agg}}) + b(p^Tl_{\text{agg}} - p_{\text{ref}}^Tl_{\text{agg}})^2 \label{eq: ofo_formulation1} \\
\text{subject to} \quad &  p_{\text{min}} \leq p \leq p_{\text{max}}, \label{eq: ofo_formulation2} \\
& 0 \leq l_{\text{agg}} \leq l_{\text{max}}, \label{eq: ofo_formulation3}  \\
& l_{\text{agg}}=\mathop{\sum}_{i=0}^{M}l_i, \quad l_i \in S(p),  \label{eq: agg}
\intertext{\noindent with lower-level optimization}
 S(p) :=& \arg \mathop{\min}_{l_i'\in\mathbb{R}^{96}}   \mathop{\sum_{i=0}^M}p^Tl_i' + 0.01\mathop{\sum_{i=0}^M}l_i'^Tl_i', \label{eq: ofo_formulation5}\\
 \text{subject to}& \quad  c_i^T l_i'=d_i, \quad i = 0, \ldots,M, \label{eq: ofo_formulation6} \\
& \quad 0 \leq l_i' \leq l_{\text{max}, i} , \quad i = 0, \ldots,M. \label{eq: ofo_formulation7}
\end{align}
\label{eq: ofo_formulation}
\end{subequations}

The DSO-level problem is defined by Eq.~\eqref{eq: ofo_formulation1}--\eqref{eq: ofo_formulation3}. Its
decision variable is the day-ahead price vector $p\in\mathbb{R}^{96}$, and the parameter is
the aggregate load $l_{\mathrm{agg}}$, measured from the current day and induced by the EV response to
prices. 
\jr{The DSO objective is to reduce peak demand while limiting deviations from the reference tariff.} Specifically, the objective combines a peak-load term $\max(l_{\mathrm{agg}})$, and a cost-deviation term $(p^\top l_{\mathrm{agg}}-p_{\mathrm{ref}}^\top l_{\mathrm{agg}})^2$ with $b$ as a weighting factor. The latter penalizes large deviations in total electricity cost relative to the reference tariff $p_{\mathrm{ref}}  \in \mathbb{R}^{96}$, thereby discouraging price designs that achieve peak reduction by imposing excessive additional cost on consumers.
The box constraints on $p$ enforce practical
price limits, while the bounds on $l_{\mathrm{agg}}$ ensure that the resulting aggregate load remains
within network capacity limits.
Although $l_{\mathrm{agg}}$ is not a direct decision variable of the DSO, it appears in the constraints because it is determined by the EV response to the price signal. Equality constraints are not included in the present formulation, but they can be added when needed to represent additional system requirements, such as exact power balance.

The price-responsive agent problem is defined by Eq.~\eqref{eq: ofo_formulation5}--\eqref{eq: ofo_formulation7}.
For each EV, the parameter is the price vector $p$ for the current day, and the decision variable is the optimal charging
profile $l_i$. Each EV is characterized by an energy demand, charging-rate limits, and an
availability window. 
\jr{The agent objective in Eq.~\eqref{eq: ofo_formulation5} minimizes the charging cost together with a quadratic regularization term. This term guarantees a unique optimal charging profile from a theoretical perspective. In practice, it discourages excessively concentrated charging power and can represent grid-side requirements for smoother demand profiles.}
The
equality constraint (Eq.~\eqref{eq: ofo_formulation6}) ensures that the required charging energy is delivered, while the inequality box
constraint (Eq.~\eqref{eq: ofo_formulation7}) enforces safe charging limits. Since the objective and constraints are convex and
separable across EVs, the joint optimization of all EVs is equivalent to solving the individual EV
charging problems independently.

\section{Methodology}
\label{sec:method}

Online Feedback Optimization (OFO) is a control strategy that steers dynamic systems toward their optimal steady-state operating points using real-time feedback~\cite{hauswirth_optimization_2024}. Instead of relying on detailed models or forecasts, OFO iteratively adjusts control inputs based on measured system outputs, enabling adaptation to disturbances and uncertainties in real time. Prior studies have demonstrated that OFO can achieve convergence to optimal solutions while respecting system constraints, even under nonlinear and complex dynamics. Various implementations exist, ranging from standard to data-driven and fully model-free formulations~\cite{hauswirth_optimization_2024,simonetto2020time}. 
\jr{Consider a general unconstrained optimization problem:
\begin{align}
    \min_{u,y} \quad & \Phi(y) \nonumber \\
    \text{subject to} \quad & y = h(u) + d, \nonumber
\end{align}
where $u$ denotes the input, $y$ denotes the output, $d$ is an additive disturbance, and $h$ represents the system input--output mapping. In standard OFO, the input is updated iteratively as
\begin{equation}
    u_{k+1}
    = u_k - \eta \nabla h(u_k)^\top \nabla \Phi(h(u_k)+d)^\top,
    \nonumber
\end{equation}
where $\eta$ is the step size, and $\nabla h(u_k)$ is the sensitivity matrix obtained from applying the chain rule to the composite objective $\Phi(h(u_k)+d)$. 
This formulation shows that OFO does not require a full plant model, but still relies on local sensitivity information that describes how input variations affect the measured output. In online EV pricing, this sensitivity is generally unknown and difficult to model because the aggregate load response depends on heterogeneous user behavior and individual charging optimizations. Moreover, practical pricing problems include operational constraints, such as price bounds and aggregate load limits, which are not captured by the unconstrained formulation above.

These challenges motivate the use of online sensitivity learning, a data-driven method which estimates the required sensitivity matrix directly from observed input--output data, together with linearized output projection, which efficiently enforces operational constraints through a locally linearized approximation of the feasible output set. This implementation requires no prior model information and efficiently enforces operational constraints, making it well suited for online EV pricing design.
}

\subsection{OFO Framework}
\subsubsection{Online Sensitivity Learning}
\label{sec:osl}

In this work, we adopted Online Sensitivity Learning (OSL) algorithm~\cite{doming_online_2023} to estimate the sensitivity matrix in real time using system measurements. However, the sensitivity matrix learned from OSL is also an approximation since it relies on a linear representation between price and aggregate load. 

Under the day-ahead pricing scheme, the entire EV fleet is optimized daily,  with $p_k \in \mathbb{R}^n$ representing the price vector for day $k$ and $l_{\text{agg},k} \in \mathbb{R}^n$ denoting the aggregate load for day $k$. Let $\Delta p_k = p_{k+1} - p_k$ and $\Delta l_k = l_{\text{agg},k+1} - l_{\text{agg},k}$. OSL operates based on the following \jr{local approximation of the aggregate price-response relation}:
\begin{equation}
    \Delta l_k \approx \nabla h(p_k) \Delta p_k,
\end{equation}
where $\nabla h(p_k) \in \mathbb{R}^{n \times n}$ is the sensitivity matrix at day $k$. 
Let $s_k = \text{vec}(\nabla h(p_k)) \in \mathbb{R}^{n^2}$ and $P_{\Delta,k} = \Delta p_k^T \otimes I \in \mathbb{R}^{n \times n^2}$, with vec and $\otimes$ denoting the column-wise vector representation and the Kronecker product, respectively. The dynamic equations of the sensitivity $s_k$ and the measurement equation are given by:
\begin{subequations}
    \begin{align}
        & s_{k+1} = s_k + w_{p, k}, \label{eq:process_noise}\\
        & \Delta l_k = P_{\Delta,k}s_k + w_{m, k}, \label{eq:measure_noise}
    \end{align}
    \label{eq:learn_h}
\end{subequations}%
\hspace{-0.5em} where $w_{p, k}\in \mathbb{R}^{n^2}$ and $w_{m,k}\in \mathbb{R}^n$ are the process and measurement error, respectively. As in \cite{picallo_adaptive_2022a}, $w_{p, k}$ and $w_{m,k}$ are modeled by Gaussian noises with covariance matrices $\Sigma_{p, k}\in \mathbb{R}^{n^2 \times n^2}$ and $\Sigma_{m, k}\in \mathbb{R}^{n \times n}$ respectively. Then learning the sensitivity matrix as shown in Eq.~\eqref{eq:learn_h} can be interpreted as a classical state estimation problem. A common approach is to use recursive least squares methods, which efficiently update the sensitivity estimate as new data becomes available. In this work, we employ the strategy proposed by \citeauthor{picallo_adaptive_2022a}, which leverages a Kalman Filter formulation to improve robustness to noise and adapt to possible time-varying system dynamics~\cite{picallo_adaptive_2022a}. The update procedure is shown in Eq.~\eqref{eq:osl}.
\begin{subequations}
    \begin{align}
        \hat{s}_{k+1} &= \hat{s}_k + K_k(\Delta l_k - P_{\Delta, k}\hat{s}_k), \\
        \Sigma_{k+1}& = (I - K_kP_{\Delta ,k})\Sigma_k + \Sigma_{p, k}, \\
        K_k &= \Sigma_kP_{\Delta,k}^T(\Sigma_{m, k} + P_{\Delta ,k}\Sigma_kP_{\Delta, k}^T).
    \end{align}
    \label{eq:osl}
\end{subequations}
\hspace{-0.5em}In this formulation, $K_k \in \mathbb{R}^{n^2 \times n}$ denotes the Kalman gain that balances trust between new measurements and prior estimates, while $\Sigma_k\in \mathbb{R}^{n^2 \times n^2}$ represents the covariance matrix capturing the uncertainty in the sensitivity estimate. 
After the update in each iteration, $\hat{s}_k$ is converted back to the matrix form $\hat{\nabla}h(p_k)$ and plugged into the updated of the controller as detailed in Section~\ref{sec:lop}.

\subsubsection{Linearized Output Projection}
\label{sec:lop}

The unconstrained OFO gradient descent step in the setting of this work is given by
$$p_{k+1} = p_k - \eta \nabla_p\Phi_{\text{DSO}}(p_k, l_{\text{agg},k}).$$
However, when output constraints $p \in \mathcal{P}$ must be enforced, this update risks violating feasibility. Linearized Output Projection (LOP) addresses this by projecting the unconstrained step onto a local approximation of the feasible set~\cite{hauswirth_optimization_2024}. Specifically, at each iteration $k$, i.e. at day $k$, a discrete-time controller as shown in Eq.~\eqref{eq:lop_controller} has been considered.
\begin{align}
    p_{k+1}=\Pi_{\mathcal{P}} [p_k+\alpha (\Sigma_{\alpha}(p_k, l_{\text{agg,k}}) + \omega_{u, k})],
    \label{eq:lop_controller}
\end{align}
where $\mathcal{P}$ is a box constraint as shown in Eq.~\eqref{eq: mathcalP}, $\alpha > 0$ is a step-size, $\omega_{u, k} \sim \mathcal{N}(0, \sigma_u^2 \mathbb{1})$ is a persistent excitation and $\Sigma_{\alpha}(p, l_{\text{agg}}) \in \mathbb{R}^n$ is set as the solution to 
\begin{subequations}
    \begin{align}
        \underset{w \in \mathbb{R}^n}{\text{minimize}} \quad &\| w+H(p)\nabla\Phi_{\text{DSO}}(p, l_{\text{agg}})^T \|^2 \\
        \text{subject to} \quad &  p_{\text{min}} \leq p+\alpha w \leq p_{\text{max}}, \label{eq: mathcalP} \\
                                & 0 \leq l_{\text{agg}}+\alpha \nabla h(p)w \leq l_{\text{max}},
    \end{align}
    \label{eq:lop_equation}
\end{subequations}
\hspace{-0.5em}where $H(p) = [\mathbb{I}_n \quad \nabla h(p)^T]$. For unconstrained problems, the ideal gradient would be $w_{\text{ideal}} = - H\nabla\Phi_{\text{DSO}}^T$. Therefore, in this constrained situation, the goal is to minimize the difference between the constrained gradient $w$ and the ideal gradient $w_{\text{ideal}}$, i.e. to minimize $\| w - w_{\text{ideal}}\|=\| w+H\nabla\Phi_{\text{DSO}}^T \|$, while ensuring $p$ and $l$ remain within their respective feasible sets.

\subsubsection{OFO for Electrical Distribution Systems Price Design}

The complete control loop is summarized in Algorithm~\ref{alg:ofo}. At each time step, the system updates its internal sensitivity estimate $H_k$ and computes the next price signal $p_{k+1}$ based on the newly observed aggregate load $l_{\text{agg}, k+1}$. The algorithm is initialized with two initial price-load pairs, for $\Delta p_0$ and $\Delta l_0$, an initial sensitivity estimate $\nabla H_0$, and an initial covariance matrix $\Sigma_0$ for the first update of the sensitivity matrix. During the optimization loop, each iteration consists of two sequential steps: (i) updating the sensitivity matrix $H_{k+1}$ using the latest measurement via OSL as mentioned in Section~\ref{sec:osl}, and (ii) computing a new feasible price $p_{k+1}$ using LOP as mentioned in Section~\ref{sec:lop}. 
The loop continues until a predefined stopping criterion is met. In practice, this can be based on a fixed time horizon like $K$ days in the algorithm, a convergence condition on price or load changes (e.g., $\|p_{k+2} - p_{k+1}\| < \epsilon$), or system-level performance indicators (e.g., cost reduction, tracking error). This flexibility allows the algorithm to be tailored to operational or regulatory requirements. The overall scheme is adaptive, model-free, and supports constraint-aware price design in online settings under limited data access.

\begin{algorithm}[ht]
\footnotesize
\caption{Online Feedback Optimization for Electricity Price Design}
	\begin{algorithmic}[1]
    \Statex \textbf{Initialization (Random Input Required):} 
    \State Provide initial prices ($p \in \mathbb{R}^{n}$): $p_0$ and $p_1$, typically derived from existing price data.
    \State Collect measurements ($l_{\text{agg}} \in \mathbb{R}^{n}$): $l_{\text{agg}, 0}$ and $l_{\text{agg}, 1}$.
    \State Initialize sensitivity matrix $H_0 \in \mathbb{R}^{n \times n}$ and its covariance $\Sigma_0 \in \mathbb{R}^{n^2 \times n^2}$.
    \Statex \rule{\linewidth}{0.4pt}
    \Statex \textbf{Online Price Design}
\For{$k\in[0, K]$}
      
     \State\label{alg_line:osl} 
     Update sensitivity matrix $H$ with $\Delta p_k = p_{k+1} - p_k$, $\Delta l_k = l_{\text{agg},k+1} - l_{\text{agg},k}$, $P_{\Delta,k} = \Delta p_k^T \otimes I$ and
     \[
    \begin{aligned}
        &\hat{s}_k = [\, H_k(1,1), H_k(2,1), \dots, H_k(m,1), H_k(1,2), H_k(2,2), \dots \,]^T, \\
        &K_k = \Sigma_kP_{\Delta,k}^T(\Sigma_{m, k} + P_{\Delta ,k}\Sigma_kP_{\Delta, k}^T),\\
        &\Sigma_{k+1} = (I - K_kP_{\Delta ,k})\Sigma_k + \Sigma_{p, k} ,\\
        &\hat{s}_{k+1} = \hat{s}_k + K_k(\Delta l_k - P_{\Delta, k}\hat{s}_k) ,\\
        &H_{k+1} = \text{Reshape}(\hat{s}_{k+1}, \text{shape}(H_k)).
    \end{aligned}
     \]

     \State \label{alg_lin:excit}
     Sample the excitation noise $\omega_{u, k+1} \sim \mathcal{N}(0, \sigma_u^2 \mathbb{1})$.

     \State \label{alg_lin:lop}
     Update price $p$ with 
     \[
     \begin{aligned}
        &p_{k+2} = \Pi_{\mathcal{P}} [p_{k+1} + \alpha \Sigma_{\alpha}(p_{k+1}, l_{\text{agg,k+1}})+\omega_{u, k+1}],\\
    \text{where} \quad &\\
      & \Sigma_{\alpha}(p_{k+1}, l_{\text{agg,k+1}}) \\
        & = \arg \underset{w \in \mathbb{R}^n}{\min} \| w+(p_{k+1})\left[\mathbb{I}^n \text{ } H_{k+1}^T\right]\nabla\Phi_{\text{DSO}}(p_{k+1}, l_{\text{agg}, k+1})^T \|^2 \\
         \text{subject to}\quad &   p_{\text{min}} \leq p_{k+1}+\alpha w \leq p_{\text{max}},  \\
                                &   0 \leq l_{\text{agg},k+1}+\alpha H_{k+1}w \leq l_{\text{max}}.
    \end{aligned}
     \]
      
    \EndFor
	\end{algorithmic}
\label{alg:ofo}
\end{algorithm}

\subsection{Benchmark Method}

Existing state-of-the-art methods usually require access to individual user-level data and are therefore not applicable under the limited data-access setting considered in this work. To assess the best achievable performance under full information, we temporarily relax this limitation and construct a benchmark that approximates the minimum attainable peak demand. Specifically, the interaction between the DSO and end users is modeled as a Stackelberg game, in which the DSO acts as the leader by setting prices while anticipating the users’ optimal load responses. This yields a bilevel optimization problem, which is reformulated as a Mathematical Program with Equilibrium Constraints (MPEC) by embedding the users’ Karush–Kuhn–Tucker (KKT) conditions into the DSO optimization problem. This benchmark is then used to evaluate how closely the proposed OFO method approaches the best achievable performance under limited observability.

Mathematically, the Stackelberg game can be formulated as a bi-level optimization problem, as shown in Eq.~\eqref{eq: game_formulation}. All variables, objectives, and constraints are consistent with those in Eq.~\eqref{eq: ofo_formulation}, with one key distinction: the upper-level agent, i.e. the DSO, has full knowledge of the individual distributed loads $l_i$.

\begin{figure}[ht]
\centering
\begin{subequations}
\begin{align}
\text{Upper level} \quad \mathop{\min}_{p\in\mathbb{R}^{96}} \quad & \max( \mathop{\sum}_{i=0}^{i=M}l_i) + b(p^T \mathop{\sum}_{i=0}^{i=M}l_i - p_{\text{ref}}^T \mathop{\sum}_{i=0}^{i=M}l_i)^2 \\
\text{subject to} \quad &  p_{\text{min}} \leq p \leq p_{\text{max}} , \\
& 0 \leq  \mathop{\sum}_{i=0}^{i=M}l_i \leq l_{\text{max}} , \\
 \text{Lower level} \quad &l_i = \arg \mathop{\min}_{l_i'\in\mathbb{R}^{96}} \mathop{\sum_{i=0}^M}p^Tl_i' + 0.01\mathop{\sum_{i=0}^M}l_i'^Tl_i', \\
 \text{subject to}& \qquad  c_i^T l_i'=d_i, \quad i = 0, \ldots,M , \\
& \qquad 0 \leq l_i' \leq l_{\text{max}, i} , \quad i = 0, \ldots,M.
\end{align}
\label{eq: game_formulation}
\end{subequations}
\end{figure}

A standard method to solve this bi-level problem is to substitute the lower level problem with its KKT condition~\cite{ye_investigating_2018} and then the original bi-level problem is turned into an MPEC as shown in Eq.~\eqref{eq: MPEC_bigM}. 

\begin{figure}[ht]
\centering
\begin{subequations}
\begin{align}
\text{MPEC:} \quad \mathop{\min}_{p\in\mathbb{R}^{96}, l_i\in\mathbb{R}^{96}} \quad & \max\Big( \mathop{\sum}_{i=0}^{i=M}l_i\Big) + b\Big|p^T \mathop{\sum}_{i=0}^{i=M}l_i-p_\text{ref}^T \mathop{\sum}_{i=0}^{i=M}l_i\Big|  \\
\text{subject to} \quad &  p_{\text{min}} \leq p \leq p_{\text{max}} , \\
& 0 \leq  \mathop{\sum}_{i=0}^{i=M}l_i \leq l_{\text{max}} , \\
 &  p+\lambda_ic_i^T +0.02l_i- \mu_i+\gamma_i=0, \\
 & \qquad \qquad \quad \forall i \in [M], \nonumber \\
&  c_i^Tl_i=d_i, \quad \forall i \in [M], \\
&  0 \leq l_i \leq l_{\text{max}, i}, \quad \forall i \in [M] , \\
&  \mu_i \geq 0, \gamma_i \geq 0 \quad \forall i \in [M], \\
&  \mu_{i, j} \leq Bz_{i, j}, \quad \forall i \in [M], \forall j \in [n] , \\
& l_{i, j} \leq B(1-z_{i, j}), \quad \forall i \in [M], \forall j \in [n] ,\\
&  z_{i, j} \in \{0, 1\}, \quad \forall i \in [M], \forall j \in [n], \\
&  \gamma_{i, j} \leq Bx_{i, j}, \quad \forall i \in [M], \forall j \in [n] ,\\
&  -(l_{i, j}-l_{max, i, j}) \leq B(1-x_{i, j}),\\
& \qquad \qquad \quad \forall i \in [M], \forall j \in [n], \nonumber \\
&  x_{i, j} \in \{0, 1\}, \quad \forall i \in [M], \forall j \in [n] .
\end{align}
\label{eq: MPEC_bigM}
\end{subequations}
\end{figure}

Here, $\lambda_i$, $\mu_i$, and $\gamma_i$ are the dual variables associated with the constraints
$c_i^\top l_i=d_i$, $0 \leq l_i$, and $l_i \leq l_{\max,i}$, respectively. Since each lower-level
problem is strongly convex and satisfies Slater’s condition, the KKT conditions are necessary and
sufficient for optimality. However, the complementary slackness conditions introduce non-convex
bilinear terms through products of primal and dual variables.

\jr{To obtain a tractable single-level reformulation, we apply the Big-$M$ method~\cite{jose_representation_1981} and let $B$ denote a sufficiently large constant.} The complementarity conditions are then replaced by
binary variables and Big-$M$ constraints, so that either the dual variable or the corresponding
primal constraint is active. This avoids the original bilinear complementarity terms and yields the
MPEC formulation in Eq.~\eqref{eq: MPEC_bigM}. When valid bounds are available for the primal and
dual variables, and the upper-level problem is otherwise unchanged, this reformulation is equivalent
to Eq.~\eqref{eq: game_formulation}~\cite{jose_representation_1981}.

Although the KKT-based reformulation removes the complementarity products, the resulting single-level
problem still contains nonlinear terms in the DSO objective arising from the interaction between
price and load. In particular, the original cost-deviation penalty involves the square of the total
cost difference, which becomes highly nonlinear when both price and load are decision variables. In
principle, such terms could be approximated using McCormick envelope relaxations~\cite{mitsos_mccormick_2009}.
However, in our setting, the feasible load range is wide, which makes accurate approximation
difficult even with multiple partitions.
To improve numerical tractability, we replace this nonlinear cost penalty by a linear absolute-value
penalty in the MPEC objective. This preserves the role of the term as a penalty on cost deviation,
while yielding a formulation that can be handled more efficiently by standard optimization solvers.
The coefficient $b$ is kept consistent with the OFO formulation to maintain comparable weighting
between peak reduction and cost deviation. 

Importantly, this modification has limited effect on the
practical solution because the cost-deviation term remains a secondary objective, while peak
reduction remains the dominant driver of the pricing signal. Since the feasible set and operational
constraints are unchanged, the resulting optimal price profiles and aggregate load trajectories show
only negligible differences in our numerical experiments. Commercial solvers such as GUROBI~\cite{gurobi}
can then solve the resulting formulation efficiently.

\subsection{Theoretical Proof of Convergence}
\label{sec:proof_converge}

Here we analyze the convergence properties of the proposed OFO-based price design approach. To enable the analysis, we first introduce a smooth reformulation of the original problem, which provides the regularity required for stochastic approximation arguments. We then show that, under standard assumptions, the online sensitivity learning component is asymptotically unbiased with bounded covariance, and use this result to establish that the projected OFO price updates converge to the stationary set of the smoothened problem. This demonstrates that the overall method asymptotically attains the desired constrained stationary behavior under limited feedback.

\subsubsection{A Smooth Reformulation of the Original Problem}

To facilitate the theoretical analysis, a smoothened version of OFO is first introduced as shown in
Eq.~\eqref{eq: smooth_formulation}. \jr{This formulation replaces non-smooth objective terms with smooth approximations and incorporates hard constraints into the objective as smooth penalty terms, enabling the use of standard tools from stochastic
approximation and convex analysis.} Specifically, in the objective function
Eq.~\eqref{eq: 1ofo_formulation1}, the $\max(\cdot)$ operator is replaced by the log-sum-exp
smoothing $\mathrm{LSE}_\tau(\cdot)$, i.e.,
$\mathrm{LSE}_\tau(z)=\tau\log\!\big(\sum_i \exp(z_i/\tau)\big)$, which is a classical smooth
approximation of $\max_i z_i$ with a controlled approximation gap and a Lipschitz-continuous
gradient for fixed $\tau>0$ \cite{boyd2004convex}. This
Lipschitz-gradient property is essential for establishing convergence of the OFO price updates
under projected stochastic approximation.

The second modification concerns the EV charging problem. The original equality and inequality
constraints are incorporated into the lower-level objective via smooth penalty terms, which yields
a continuously differentiable lower-level problem. This differentiability is critical for
guaranteeing convergence of the online sensitivity learning process based on gradient-type
updates. Such smooth penalty reformulations are standard in constrained optimization for enabling
differentiable surrogates and gradient-based analysis \cite{nocedal2006numerical}.

\textbf{All subsequent convergence analysis is carried out based on this smoothened
formulation}.

\begin{figure}[ht]
\begin{subequations}
    \centering
\begin{align}
\text{Smooth OFO:}& \nonumber \\
\mathop{\min}_{p\in\mathbb{R}^{96}} &  \quad\text{LSE}(l_\text{agg}) + b(p^Tl_{\text{agg}} - p_{\text{ref}}^Tl_{\text{agg}})^2 \label{eq: 1ofo_formulation1} \\
\text{subject to} \quad &  p_{\text{min}} \leq p \leq p_{\text{max}} ,\label{eq: 1ofo_formulation2} \\
& 0 \leq l_{\text{agg}} \leq l_{\text{max}}, \label{eq: 1ofo_formulation3}  \\
& l_{\text{agg}} = \mathop{\sum}_{i=0}^{i=M}l_i , \label{eq: 1agg}\\
& l_i = \arg \min \sum_{i=0}^{M} (
p^\top l_i
+ 0.01\, l_i^Tl_i \label{eq: 1ofo_formulation7}\\
& \quad+ \rho_{\mathrm{eq}}\, ( c_i^\top l_i - d_i)^2
)+ \rho_{\mathrm{ineq}}
\sum_{i=0}^{M} (
\|\max(0,-l_i)\|_2^2 \nonumber  \\
& \quad+ \|\max(0,\,l_i - l_{\max,i})\|_2^2
). \nonumber
\end{align}
\label{eq: smooth_formulation}
\end{subequations}
\end{figure}

A corresponding smoothened MPEC can be derived in the same manner as the non-smooth MPEC introduced earlier, and for brevity, its explicit formulation is omitted. We refer to these formulations as \textit{smooth OFO} and \textit{smooth MPEC}, respectively, and our goal here is to establish that the smooth OFO algorithm converges to the solution set of the smooth MPEC.

\subsubsection{Convergence of OFO Results}
\label{sec: OFO_converge}

Let $\mathcal P \subset \mathbb{R}^m$ denote the compact and convex feasible set for prices $p$
(corresponding to Eq.~\eqref{eq: 1ofo_formulation2}). For each $p \in \mathcal P$, consider the
agent-level charging problem with QP-penalty objective $\Phi_{\mathrm{Agent}}(p,l)$ (corresponding to Eq.~\eqref{eq: 1ofo_formulation7}). Due to the
strong convexity of $\Phi_{\mathrm{Agent}}(p,l)$ in the smoothened OFO formulation, the lower-level
problem admits a unique optimizer, denoted by
\[
l^\star(p) := \arg\min_{l \in \mathcal \mathbb{R}^m} \Phi_{\mathrm{Agent}}(p,l),
\]

Consequently, the smooth MPEC is equivalent to the reduced single-level problem
\begin{subequations}
\begin{align}
\min_{p \in \mathcal{P}} \quad & \Phi_{\mathrm{DSO}}(p) := F\bigl(p, l^\star(p)\bigr), \\
\text{s.t.}\quad & l^\star(p) = \arg\min_{l \in \mathbb{R}^m}\Phi_{\mathrm{Agent}}(p, l),
\end{align}
\label{eq:reduced_problem}
\end{subequations}
and any solution of Eq.~\eqref{eq:reduced_problem} corresponds to a solution of the smooth MPEC.
The optimality condition for the smooth MPEC (equivalently, for the reduced problem) is
\begin{equation}
0 \in \nabla \Phi_{\mathrm{DSO}}(p^\star) + N_{\mathcal P}(p^\star),
\label{eq:smooth_MPEC_result}
\end{equation}
where $N_{\mathcal P}(p^\star)$ denotes the normal cone to $\mathcal P$ at $p^\star$.

The smoothened OFO price update
can be written as the projected stochastic recursion
\begin{equation}
\label{eq:ofo_projected}
p_{t+1}
=
\Pi_{\mathcal P}\!\left(
p_t-\alpha_t\,\widehat g_t
\right),
\end{equation}
where $\widehat g_t$ is the estimated gradient of the reduced DSO objective with respect to the
price vector. We now show that the smoothened OFO iterates converge to the same constrained stationary set.

\begin{lemma}[Online sensitivity learning convergence~\cite{picallo_adaptive_2022}]
\label{lem:osl_convergence}
\leavevmode\\
Under the lower-level optimization in Eq.~\eqref{eq: 1ofo_formulation7}, with the process and measurement noise models in Eqs.~\eqref{eq:process_noise}--\eqref{eq:measure_noise}, the feasible set defined in Eqs.~\eqref{eq: 1ofo_formulation2}--\eqref{eq: 1ofo_formulation3}, and the persistent excitation condition in Eq.~\eqref{eq:lop_controller}, the associated online sensitivity estimate is asymptotically unbiased and has bounded covariance. Specifically,
\[
\bigl\|\mathbb{E}[h_t-\hat h_t]\bigr\|_2^2 \le C_{h,1}e^{-C_{h,2}t}\xrightarrow[t\to\infty]{}0,
\]
and
\[
\mathbb{E}\!\left[\|h_t-\hat h_t\|_2^2\right]
=\operatorname{tr}(\Sigma_t)
\le C_{h,3}+C_{h,4}e^{-C_{h,5}t}\xrightarrow[t\to\infty]{}C_{h,3},
\]
for some constants $C_{h,i}>0$, where $\Sigma_t$ denotes the estimation error covariance.
\end{lemma}

\begin{proof}
For the smoothened EV charging problem considered here, the lower-level objective of the smoothened OFO is continuously
differentiable and strongly convex, with a Lipschitz-continuous gradient due to the quadratic
regularization and squared penalty terms. This implies that the associated operator is strongly
monotone and Lipschitz continuous in certain regions around nominal operating points~~\cite{picallo_adaptive_2022}. Since $p$ and $l$ are physically constrained, the iterates remain
in a compact set and the relevant gradients and Jacobians are uniformly bounded. In addition,
persistent excitation of the price signal ensures a uniform positive lower bound on the regressor
covariance matrices, while bounded variations of $p$ and $l$ imply a uniform upper bound. Hence, under these conditions, the estimate is asymptotically unbiased and has bounded covariance.
\end{proof}

\begin{assumption}[Summable step size]
\label{ass:stepsize}
\leavevmode\\
The step sizes $\{\alpha_t\}$ satisfy the Robbins--Monro conditions:
$\alpha_t>0$, $\sum_{t=0}^{\infty}\alpha_t=\infty$, and $\sum_{t=0}^{\infty}\alpha_t^2<\infty$.
\end{assumption}

The step-size condition in Assumption~\ref{ass:stepsize} is standard in stochastic approximation and commonly appears in related online optimization and feedback-based learning methods to ensure convergence despite noisy updates~\cite{Karandikar_convergence_2024}. A simple example is
\begin{equation}\label{eq:step_size_decay}
    \alpha_t=\frac{\alpha_0}{t+1}, \qquad \alpha_0>0.
\end{equation}
\jr{The above condition implies that the step size asymptotically converges to zero. This specification facilitates the subsequent theoretical analysis and characterization. In practice, however, practitioners may apply the controller over consecutive periods, where the step size in each period follows \eqref{eq:step_size_decay}, and the solution obtained in the last period is used to warm start the next. Under this implementation, the closed-loop guarantees presented in this subsection naturally extend to scenarios without a fixed termination time.
}

\begin{assumption}[Regularity of the reduced objective]
\label{ass:regularity_reduced}
\leavevmode\\
The reduced objective $\Phi_{\mathrm{DSO}}$ is continuously differentiable on $\mathcal P$, and
$\nabla \Phi_{\mathrm{DSO}}$ is Lipschitz continuous on $\mathcal P$.
The gradient oracle used by OFO satisfies
\[
\widehat{\nabla \Phi}(p_t)
=
\nabla \Phi(p_t) + \xi_t,
\]
where $\{\xi_t\}$ is an unbiased martingale-difference noise sequence with bounded variance.
\end{assumption}

The verification of Assumption~\ref{ass:regularity_reduced} follows from Lemma~\ref{lem:osl_convergence}. In particular, the lemma implies that the sensitivity estimation error has bounded second moment and is asymptotically unbiased. Since the OFO gradient surrogate is constructed from the sensitivity estimate through smooth mappings on a compact domain, these properties carry over to the induced gradient error $\xi_t$. For brevity, we defer the detailed derivation to the proof of Theorem~\ref{thm:ofo_convergence}. Note that the resulting unbiasedness holds asymptotically rather than at every iteration, which is typical in stochastic approximation analyses~\cite{kushner_stochastic_2013}.

\begin{theorem}[Convergence of smoothened OFO]
\label{thm:ofo_convergence}
\leavevmode\\
Suppose Assumptions~\ref{ass:stepsize} and~\ref{ass:regularity_reduced} hold, and the online
sensitivity learning conditions in Lemma~\ref{lem:osl_convergence} are satisfied. Then the projected
smoothened OFO iterates generated by Eq.~\eqref{eq:ofo_projected} converge almost surely to the constrained
stationary set of the reduced problem Eq.~\eqref{eq:reduced_problem}. Equivalently, every limit point
$p^\star$ satisfies
\[
0 \in \nabla \Phi_{\mathrm{DSO}}(p^\star) + N_{\mathcal P}(p^\star),
\]
which coincides with the stationary condition of the smooth MPEC in
Eq.~\eqref{eq:smooth_MPEC_result}.
\end{theorem}

\begin{proof}
Projected stochastic approximation is well studied in the literature. For compact convex feasible sets, introducing a projection step preserves the convergence properties of the Robbins--Monro recursion under the usual conditions on step size, smoothness, and noise~\cite{kushner_stochastic_2013,borowski_convergence_2025}. Therefore, as stated in Theorem.~\ref{thm:ofo_convergence}, the projected OFO update can be analyzed within the projected stochastic approximation framework, and its limit points satisfy the constrained first-order optimality condition
$0 \in \nabla \Phi_{\mathrm{DSO}}(p^\star) + N_{\mathcal P}(p^\star)$.

We first verify that the OFO update fits the Assumption.~\ref{ass:regularity_reduced} frame, as commonly used in stochastic approximation analyses~\cite{Madden_high_2024}.
Since the DSO objective depends on the price both directly and indirectly through the aggregate load
response $l_{\mathrm{agg}}(p)$, the corresponding reduced gradient is

\begin{equation}
 g_t
=
\frac{\partial \Phi_{\mathrm{DSO}}}{\partial p}(p_t,l_{\mathrm{agg},t})
+
h_t^\top
\frac{\partial \Phi_{\mathrm{DSO}}}{\partial l_{\mathrm{agg}}}(p_t,l_{\mathrm{agg},t}),   
\end{equation}

where
\[
h_t := \frac{\partial l_{\mathrm{agg}}}{\partial p}(p_t),
\]
denotes the local price--load sensitivity. Replacing the unknown sensitivity $H_t$ by its online
estimate $\hat h_t$ gives the OFO gradient estimate in Eq.~\eqref{eq:ofo_projected}
\begin{equation}
\widehat g_t
=
\frac{\partial \Phi_{\mathrm{DSO}}}{\partial p}(p_t,l_{\mathrm{agg},t})
+
\hat h_t^\top
\frac{\partial \Phi_{\mathrm{DSO}}}{\partial l_{\mathrm{agg}}}(p_t,l_{\mathrm{agg},t}).
\end{equation}
Hence,
\begin{align}
&\widehat g_t = g_t + \xi_t,
\\
&\xi_t
=
(\hat h_t-h_t)^\top
\frac{\partial \Phi_{\mathrm{DSO}}}{\partial l_{\mathrm{agg}}}(p_t,l_{\mathrm{agg},t}),
\end{align}
so the gradient estimation error is induced directly by the sensitivity estimation error. Since the feasible set is compact and $\Phi_{\mathrm{DSO}}$ is continuously differentiable, there exists
a constant $C_q>0$ such that $\|\frac{\partial \Phi_{\mathrm{DSO}}}{\partial l_{\mathrm{agg}}}(p_t,l_{\mathrm{agg},t})\|\le C_q$ for all $t$. Then,
\begin{equation}
\|\xi_t\|
\le
\|\hat h_t-h_t\|\,\|q_t\|
\le
C_q \|\hat h_t-H_t\|,
\end{equation}
which implies
\begin{equation}
\mathbb E[\|\xi_t\|^2]
\le
C_q^2\,\mathbb E[\|\hat h_t-h_t\|^2].
\end{equation}
By Lemma~\ref{lem:osl_convergence}, the sensitivity estimation error has bounded variance, which in turn implies that $\xi_t$ has bounded second moment. In addition, Lemma~\ref{lem:osl_convergence} establishes that $\hat h_t$ is asymptotically unbiased for $h_t$. Because $\xi_t$ is obtained from the sensitivity estimation error through a bounded linear transformation, $\xi_t$ is likewise asymptotically unbiased.

 Strictly speaking, Assumption~\ref{ass:regularity_reduced} is not guaranteed at every iteration by Lemma~\ref{lem:osl_convergence}. The lemma instead shows that the induced gradient error has bounded variance throughout and becomes asymptotically unbiased as $t\to\infty$. In this sense, the noise condition is satisfied in the asymptotic form relevant to the convergence argument. This is sufficient for the projected stochastic approximation analysis used here and is consistent with classical convergence results for stochastic approximation and gradient methods with diminishing errors, where asymptotically negligible bias is allowed~\cite{kushner_stochastic_2013,borowski_convergence_2025,Kushner_Stochastic_2010}.

Finally, let the step-size sequence satisfy Assumption~\ref{ass:stepsize}. Then under the properties established above, the assumptions of the projected stochastic approximation theorem are satisfied. It follows that the iterates generated by Eq.~\eqref{eq:ofo_projected} converge almost surely to the set of projected stationary points of the reduced problem.
Equivalently, every limit point $p^\star$ satisfies
\[
0 \in \nabla \Phi_{\mathrm{DSO}}(p^\star)+N_{\mathcal P}(p^\star).
\]
Since \eqref{eq:reduced_problem} is equivalent to the smooth MPEC, this is exactly the stationary condition
in Eq.~\eqref{eq:smooth_MPEC_result}.
\end{proof}

The proof relies on three steps: smoothing the original problem to obtain a differentiable MPEC, verifying the conditions required for projected stochastic approximation, and then applying a projected stochastic approximation argument. As a result, the smoothened OFO iterates converge almost surely to the stationary set of the smooth MPEC while remaining feasible at every iteration. In the context of online price design, this theorem shows that the smoothened OFO updates converge to a stationary price signal consistent with the smooth MPEC formulation.


\subsection{Efficient Implementation}
\label{sec:enhancement}

To improve the practical efficiency of the online feedback optimization loop, we introduce several enhancements that accelerate convergence and enhance sensitivity learning in complex environments without modifying the core structure of the algorithm.

\textit{\textbf{Event-Triggered Tuning: }} To improve the efficiency and informativeness of sensitivity learning, we apply event-triggered tuning by updating prices only when EV charging activity is detected. Specifically, the gradient-based update is applied conditionally, using a binary mask that activates only during relevant system behavior. The price update becomes
\begin{equation}
     p_{k+1} = p + \alpha_{\text{mask}}\alpha_k \Sigma_{\alpha}(p_k, l_{\text{agg}, k}),
\end{equation}
where $\alpha_{\text{mask}} \in \{0, 1\}^n$ indicates whether the aggregate EV load exists. When 
$\alpha_{\text{mask, i}}=0$, the price update for time step $i$ during the day is skipped, preventing uninformative or misleading updates during periods with no charging activity. This selective mechanism enhances data efficiency by ensuring that sensitivity learning concentrates on time steps that truly reflect price-responsive behavior.

\textit{\textbf{Warm Initialization of Sensitivity: }} To accelerate convergence in the early stages of online sensitivity learning, warm initialization is employed by estimating the initial sensitivity matrix from historical aggregate data. Rather than starting from a neutral or uninformative prior, the sensitivity matrix $H_0 \in \mathbb{R}^{n \times n}$ is initialized using previously observed price-load patterns, typically averaged over several past months with similar conditions. Specifically, in this work, to ensure sufficient excitation for identifying price-response behavior, small random perturbations are added to the reference price $p_{\text{ref} }\in \mathbb{R}^n$ each day. This process generates a diverse set of price-load pairs, from which a basic initial sensitivity estimate is derived. The warm-started sensitivity matrix offers an informed initial guess of the system's price-response behavior, enabling the algorithm to begin with a more accurate model. As online updates continue, this estimate is progressively refined using new measurements, allowing for faster adaptation while preserving long-term accuracy. It should be noted that, both in historical data collection and in online optimization process, only aggregate load measurement is collected and therefore, the proposed enhancement does not violate the privacy protection principle.

The above two enhancements do not modify the core structure of the OFO algorithm and are only designed to improve practical performance while preserving the fundamental update mechanism. Therefore they do not affect the convergence properties established in previous analysis.

\section{Case Study Settings}
\label{sec:settings}

In the previous sections, both the original and smoothened OFO formulations have been introduced.
In this case study, we utilize the original formulation in Eq.~\eqref{eq: ofo_formulation} to evaluate practical performance, while the smooth reformulation in Eq.~\eqref{eq: smooth_formulation} is primarily used for convergence analysis. 
Although the two formulations behave similarly in our simulations (Section~\ref{sec:comp_origin_smooth}), we report case-study results based on the original formulation because it is the one implemented in practice. The smoothened formulation is introduced as an analysis surrogate as it ensures differentiability and enables the convergence proof, but it also introduces approximation parameters and may slightly bias the objective and constraint handling. For transparency, all reported numerical results therefore correspond to the original OFO implementation, while the smoothened formulation is used to support the theoretical analysis. Besides, a constant step size is chosen empirically through trial and error to balance convergence speed, numerical stability, and implementation simplicity. This choice is common in finite-horizon implementations of stochastic approximation~\cite{kushner_stochastic_2013}. With a constant step size, convergence is often characterized in a weak sense, corresponding to stable tracking behavior and small steady-state fluctuations around the stationary set.

\subsection{Practical Implementation}
\label{sec:practical_implementation}

This study considers EV charging as the representative flexible load for optimal pricing in the power systems. EV charging is chosen because unregulated charging can significantly increase peak demand (up to 74.4\% at full penetration~\cite{Jaru_predicting_2023}) and place additional stress on grid operation. Non-flexible loads are excluded because they cannot be shifted through price signals and therefore do not contribute to demand-side flexibility. The price-responsive agent problem is defined by Eq.~\eqref{eq: ofo_formulation5}--\eqref{eq: ofo_formulation7}.

At the upper level, the DSO manages grid flexibility and maintains reliable operation under increasing load variability, particularly from EV charging. In this work, dynamic pricing is used to reduce peak demand while limiting deviations in total electricity cost from the reference tariff. In this setting, limiting cost deviation means that peak reduction should not be achieved by imposing an excessive additional cost burden on consumers. The corresponding DSO-level problem is defined in Eq.~\eqref{eq: ofo_formulation1}--\eqref{eq: ofo_formulation3}.
The nonsmooth term $\max(l_{\text{agg}})$ does not hinder implementation of the original OFO. In the CasADi-based implementation, the active maximum component is used to construct a valid subgradient, which provides a descent direction for the price update. This practical detail contributes to the strong performance of the original OFO and helps explain why the original and smooth reformulations show similar behavior in the case studies.

Further details of the model formulation are provided in Section~\ref{sec:prob}.

 \subsection{EV Charging Datasets}
 \label{sec:ev_dataset}

\begin{figure}[!t]
    \centering
    \includegraphics[width=\linewidth]{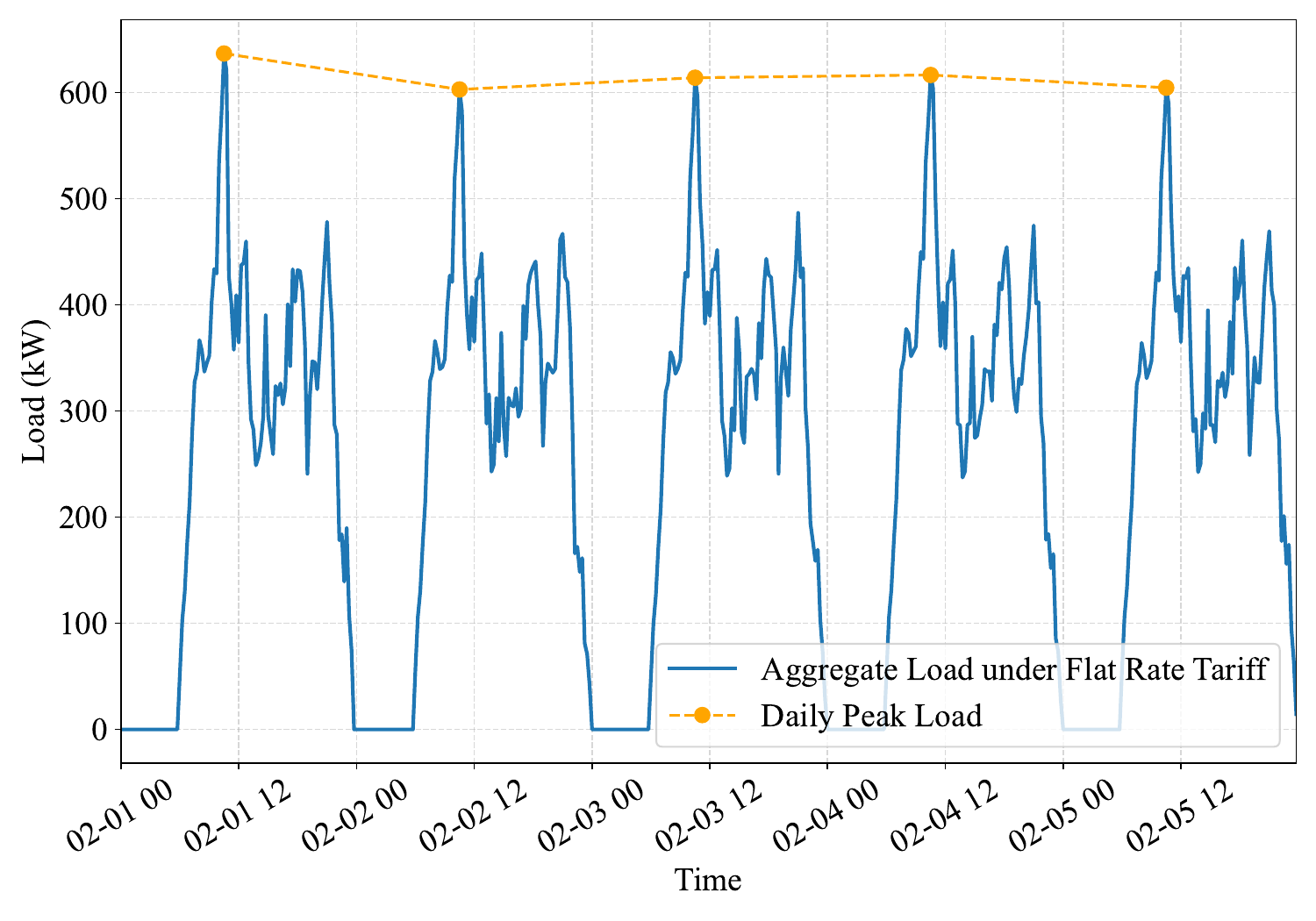}
    \caption{Visualization of Dynamic EV fleet under flat-rate tariff.}
    \label{fig:dyanmic_dataset}
\end{figure}

The proposed pricing strategies are evaluated using synthetic EV charging demand based on structured datasets that capture representative charging-session behavior. Each dataset contains the information required for time-resolved charging simulation, including session start and end times, total energy demand and the maximum charging power available at each time step.

In this work, we first consider an average daytime workplace charging scenario to represent typical arrival and departure patterns in commercial or public charging facilities. Each charging session starts and ends within the same day (i.e., no overnight charging). Simulations are performed at a 15-minute resolution, corresponding to 96 time steps per day. Specifically, 336 EVs are simulated over each day within the interval from 06:00 to 24:00, consistent with a workplace charging setting in which most EVs arrive in the morning and depart in the afternoon. Weekends are omitted to focus on a relatively stable and repeatable workplace charging pattern, thereby limiting variability due to irregular travel behavior and lower workplace occupancy. Moreover, because fewer EVs are typically present at the workplace on weekends, significant charging peaks are less likely to occur. Extending the study to weekly charging patterns (including weekends) is left for future work, since weekday--weekend transitions introduce stronger nonstationarity and may require faster adaptation of the online pricing and sensitivity-learning updates.

The charging-session structure (e.g., connection times and charging limits) is fixed across days, while the energy demand is sampled stochastically to represent day-to-day variability in travel behavior and state-of-charge conditions. 
\jr{This setup preserves key characteristics of charging-session behavior from the synthetic EV profiles while introducing demand uncertainty.}

As a result, the aggregate load profile evolves dynamically over time, reflecting both inherent demand variability and adaptive responses to price signals. To illustrate the baseline charging behavior, charging is first simulated under a flat-rate tariff without load-shaping control, as shown in Fig.~\ref{fig:dyanmic_dataset}. The figure presents the resulting aggregate load profile over a continuous 5-day simulation horizon.

\subsection{Reference Price}

\begin{figure}
    \centering
    \includegraphics[width=\linewidth]{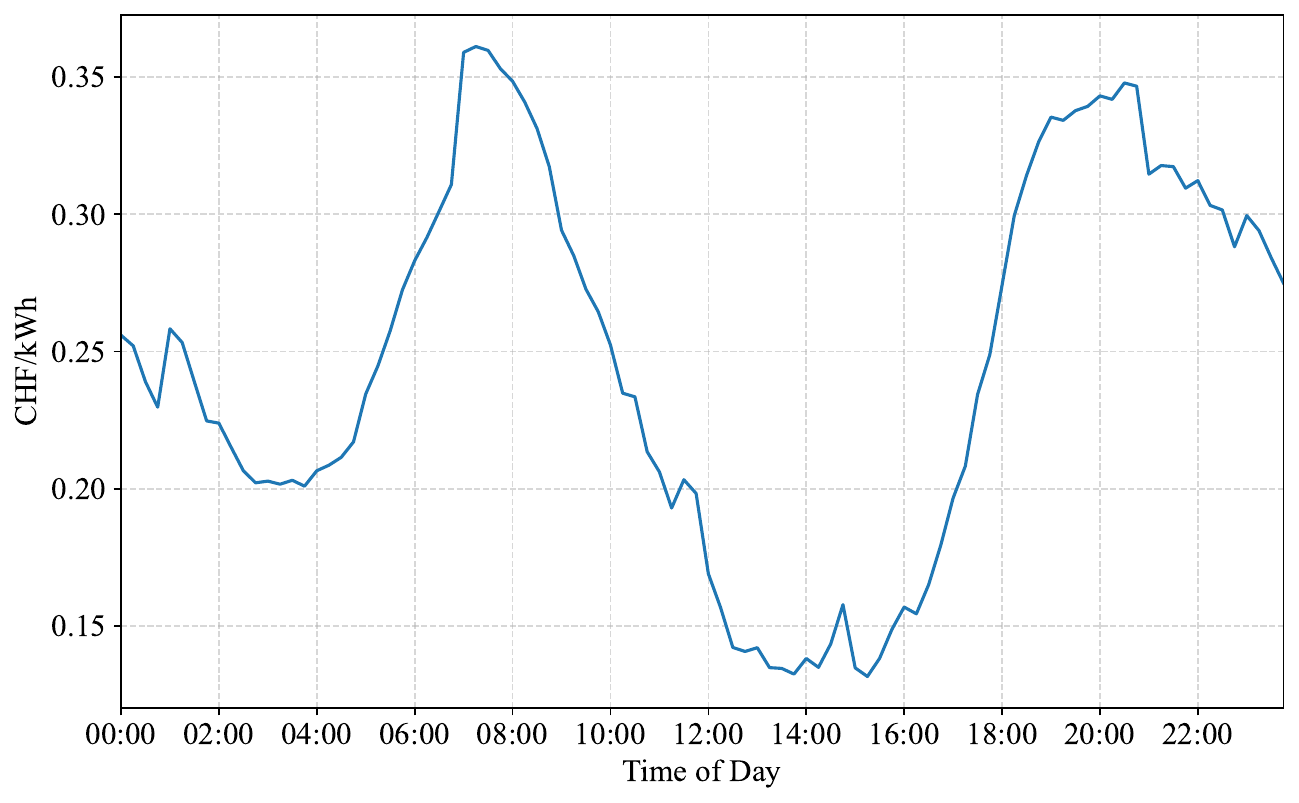}
    \caption{Reference price profile for a single day (observed on 26 May 2025).}
    \label{fig:reference_price}
\end{figure}

Fig.~\ref{fig:reference_price} shows the reference price profile $p_{\text{ref}}$ used in the experiments, based on the Vario dynamic tariff offered by Groupe E, Switzerland~\cite{groupee_vario}. The profile has two distinct high-price periods: a morning peak from approximately 07:00 to 09:00 and an evening peak from about 18:00 to 21:00, with lower prices during midday and overnight. This structure is consistent with typical Swiss demand patterns, where higher consumption is often associated with commuting hours and residential activities.

The reference price is used as a realistic baseline for performance comparison, ensuring that the results remain relevant for practical tariff design. In the experiments, it also serves as the initial price profile for all simulations, allowing a direct assessment of peak-demand reduction and load reshaping relative to a representative dynamic tariff already used in practice.

\subsection{Sensitivity Matrix for Warm Initialization of Online Sensitivity Learning}
\label{sec:historical}

In the proposed OFO method, the sensitivity matrix $H$ approximates the relationship between the price vector and the aggregate load response. Its accurate estimation is important for effective price updates, especially in the early stage when only limited feedback is available. To improve transient performance, this work uses a warm-start initialization of $H$ (Section~\ref{sec:enhancement}).

The initial matrix $H_0$ is estimated from historical aggregate-load data collected under small random day-to-day perturbations of the price profile. Let the data span $D$ days, and let $p^{(d)} \in \mathbb{R}^{96}$ and $l_{\mathrm{agg}}^{(d)} \in \mathbb{R}^{96}$ denote the price and aggregate load profiles on day $d$, respectively. Define $\Delta p_{\mathrm{hist}} \in \mathbb{R}^{(D-1)\times 96}$ and $\Delta l_{\mathrm{agg},\mathrm{hist}} \in \mathbb{R}^{(D-1)\times 96}$ as the corresponding day-to-day changes in price and aggregate load, with rows given by
\[
\Delta p_{\mathrm{hist}}^{(d-1,:)} \!=\! p^{(d)} \!-\! p^{(d-1)}, ~~
\Delta l_{\mathrm{agg},\mathrm{hist}}^{(d-1,:)} \!=\! l_{\mathrm{agg}}^{(d)} \!-\! l_{\mathrm{agg}}^{(d-1)}, \quad d=2,\dots,D.
\]
We estimate $H_0$ through the following least-squares formulation
\begin{align}
    H_0 &= \arg\min_H \left\| \Delta l_{\mathrm{agg},\mathrm{hist}} - \Delta p_{\mathrm{hist}} H^\top \right\|_F^2 \nonumber\\
    &= \left[\left(\Delta p_{\mathrm{hist}}^\top \Delta p_{\mathrm{hist}}\right)^{-1}
    \Delta p_{\mathrm{hist}}^\top \Delta l_{\mathrm{agg},\mathrm{hist}}\right]^\top,
    \label{eq:lstsqr}
\end{align}
where $\|\cdot\|_F$ denotes the Frobenius norm.

\begin{figure}[!t]
        \centering
        \includegraphics[width=\linewidth]{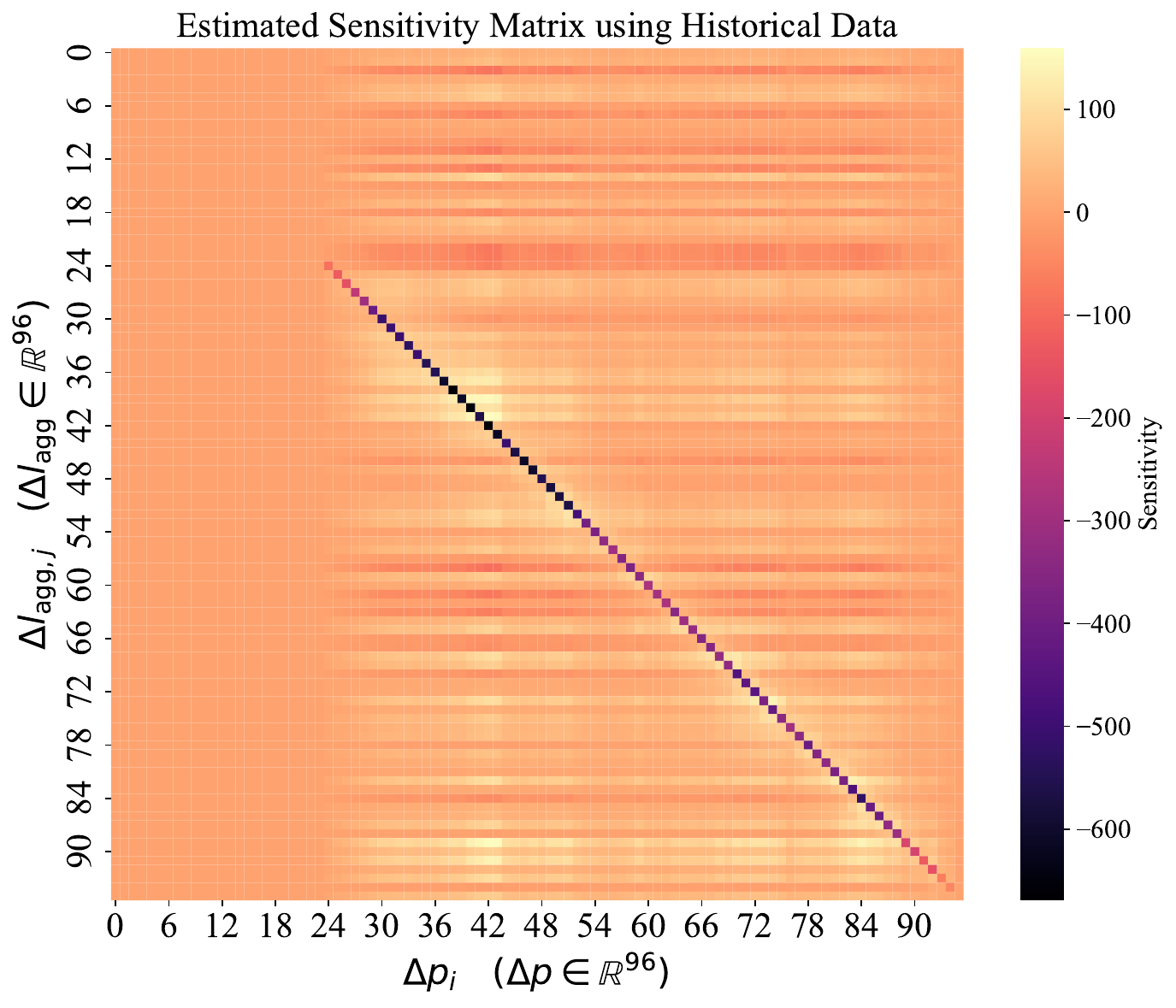}
    \caption{Estimated Sensitivity Matrix for Warm Initialization of Online Sensitivity Learning.}
        \label{fig:reallarge_H0}
\end{figure}

Fig.~\ref{fig:reallarge_H0} shows the estimated initial sensitivity matrix $H_0$ for the dataset described in Section~\ref{sec:ev_dataset}. The horizontal axis $\Delta p_i$ denotes the price difference at time step $i$ ($\Delta p \in \mathbb{R}^{96}$), and the vertical axis $\Delta l_{\mathrm{agg},j}$ denotes the aggregate load difference at time step $j$ ($\Delta l_{\mathrm{agg}}\in\mathbb{R}^{96}$). Each matrix entry quantifies the sensitivity of the aggregate load at time $j$ to a marginal change in price at time $i$.

The estimated sensitivities are concentrated around a diagonal band aligned with the active charging periods, indicating that price changes primarily affect load within or near the corresponding charging windows. Outside these periods, the entries are close to zero, which reflects the limited influence of price changes when no EVs are charging. The dominant negative values near the diagonal indicate that increasing the price at a given time step tends to reduce the aggregate charging load at that time, consistent with the expected price--demand response. Smaller positive off-diagonal entries indicate inter-temporal load shifting: load reduced at one time step is partially redistributed to other feasible charging periods to satisfy energy requirements. Overall, the structure of $H_0$ captures the expected EV charging response, namely an immediate reduction in charging power under higher prices followed by partial rescheduling within the available charging horizon.

\subsection{Experimental Settings}
\label{sec:exp_setting}

The experiments are conducted with a control and sampling resolution of 15 minutes, yielding $n=96$ time steps per day. At each step, the DSO applies a price vector $p\in\mathbb{R}^{96}$ to all agents. In response, price-responsive loads (EVs) adjust their charging schedules according to their session parameters and constraints. The aggregate load $l_{\mathrm{agg}}\in\mathbb{R}^{96}$ is obtained by summing the individual EV charging profiles and is fed back to the OFO module, which iteratively updates the price signal based on the observed response. In practice, DSO can directly measure the aggregated load at the common coupling point instead of summing the profiles. \jr{The simulation covers four months in 2020, from February to May, with weekends treated as weekdays, totaling 121 days.
}

The OFO step size is set to a constant value $\alpha=5\times 10^{-5}$, chosen empirically to ensure stable updates in a highly coupled charging environment, where small price changes may induce large aggregate-load variations. The measurement and process noise covariances are set to $\Sigma_m=\Sigma_p=100$, reflecting the dynamic simulation setting, and the initial covariance is set to $\Sigma_0=10^{3}I$ to ensure fast learning from incoming feedback. Price bounds are fixed at $p_{\min}=0.001$ and $p_{\max}=1.0$ for all datasets, while the peak-load limit $l_{\max}$ is set as $750$~kW to avoid extreme charging power. The coefficient $b$ in the OFO objective is set to $0.0002$. This small value is chosen to prioritize peak-load reduction, while the cost-deviation term acts mainly as a soft safeguard against excessively large deviations in total electricity cost relative to the reference-price case.

For the first two days, the initial price vectors are generated from the reference price profile by adding small Gaussian perturbations:
\begin{align}
p_0 &= p_{\mathrm{ref}}, \label{eq:init_price_noise_both} \\
p_1 &= p_{\mathrm{ref}} + \varepsilon_1, 
& \varepsilon_1 &\sim \mathcal{N}(0,\sigma_p^2 I),
\end{align}
with $\sigma_p=0.01$. The perturbation magnitude is kept small to avoid large transient load violations during the initial adaptation period.
\begin{figure*}[!t]
    \centering
    \includegraphics[width=0.8\linewidth]{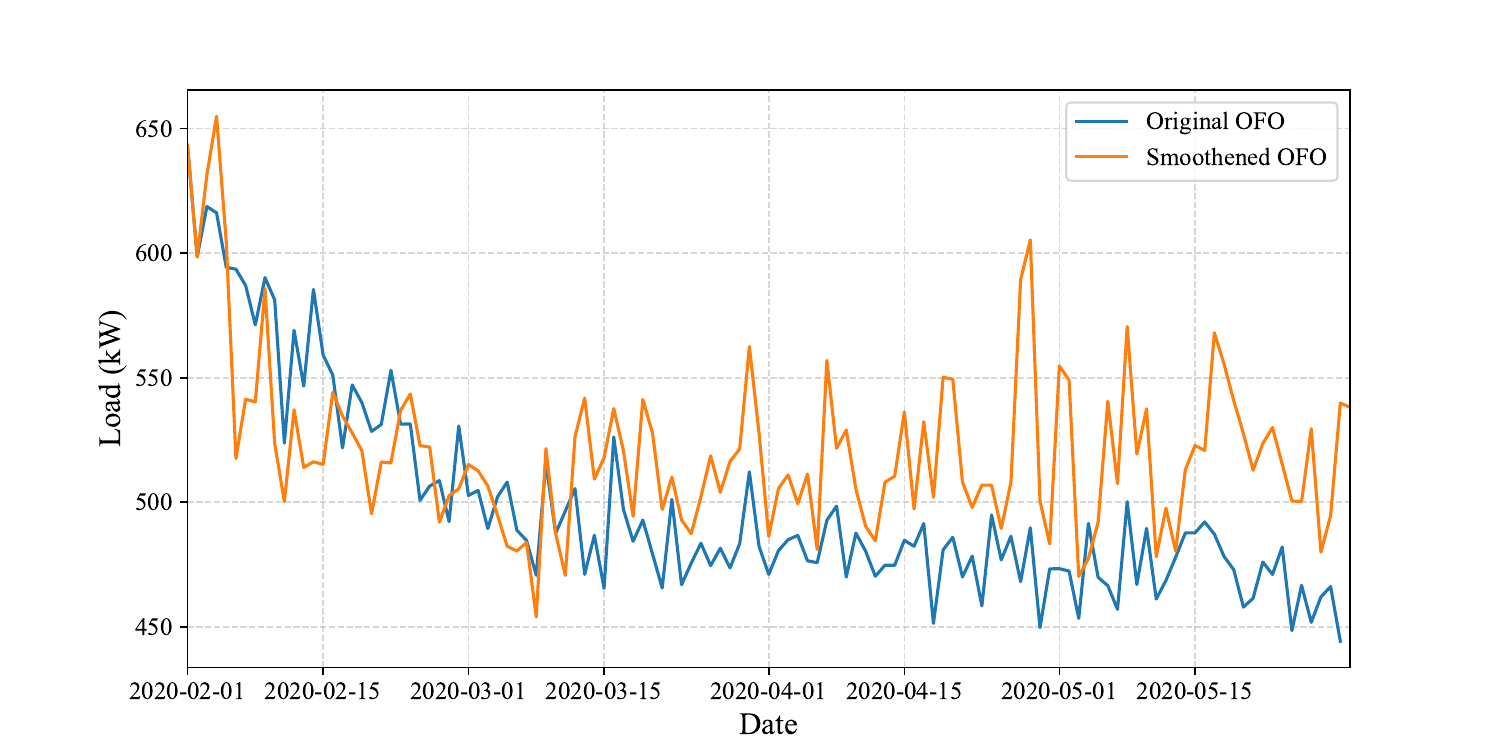}
    \caption{Comparison of Original OFO and Smoothened OFO Performances.}
    \label{fig:comp_smooth}
\end{figure*}

To illustrate robustness to initialization, we additionally simulate three more cases by applying a second small Gaussian perturbation to the initialized prices in~\eqref{eq:init_price_noise_both}:
\begin{align}
p_0^{(r)} &= p_0 + \delta_0^{(r)}, 
& \delta_0^{(r)} &\sim \mathcal{N}(0,\tilde{\sigma}_p^2 I), \\
p_1^{(r)} &= p_1 + \delta_1^{(r)}, 
& \delta_1^{(r)} &\sim \mathcal{N}(0,\tilde{\sigma}_p^2 I),
\label{eq:init_price_noise_robustness}
\end{align}
for $r=1,2,3$, with $\tilde{\sigma}_p=0.02$. The values $\sigma_p=0.01$ and $\tilde{\sigma}_p=0.02$ are randomly chosen to generate mild initialization variability without causing excessive transient constraint violations. Together with the nominal run, these four runs are used to report average performance and confidence intervals in the following section.

\jr{
To further evaluate the robustness of OFO under model mismatch, we introduce a demand-timing mismatch between the model used to compute the warm-start sensitivity matrix as well as the MPEC benchmark price, and the true agent response observed during online operation. In the original lower-level model, Eq.~\eqref{eq: ofo_formulation6} enforces demand fulfillment by requiring the cumulative charging load of EV~$i$ over its availability window to match its energy demand, where ($c_i \in \{0,1\}^{n}$) indicates the availability window of EV~$i$. In the online simulation, the true availability window is delayed by one time step. Let $S\in\mathbb{R}^{n\times n}$ denote the one-step delay operator. The true availability indicator is defined as
\begin{equation}
\tilde c_i = S c_i,
\end{equation}
and the corresponding true demand fulfillment constraint becomes
\begin{equation}
\tilde c_i^\top l_i = d_i.
\end{equation}
Hence, the DSO initializes OFO with a sensitivity matrix derived from the original availability indicator $c_i$, whereas the measured aggregate load is generated by agents following the shifted indicator $\tilde c_i$. This change shifts the feasible charging periods and therefore alters the aggregate price-response structure, making the MPEC-based price computed from the original model ($c_i$) suboptimal for the true system. The experiment is designed to demonstrate that OFO can maintain effective performance under model mismatch by using aggregate feedback to update the sensitivity matrix online, even when the warm-start sensitivity is structurally inaccurate, as shown in Section~\ref{sec:mismatch_result}.
}

\begin{figure*}[!t]
    \centering
    \includegraphics[width=0.8\linewidth]{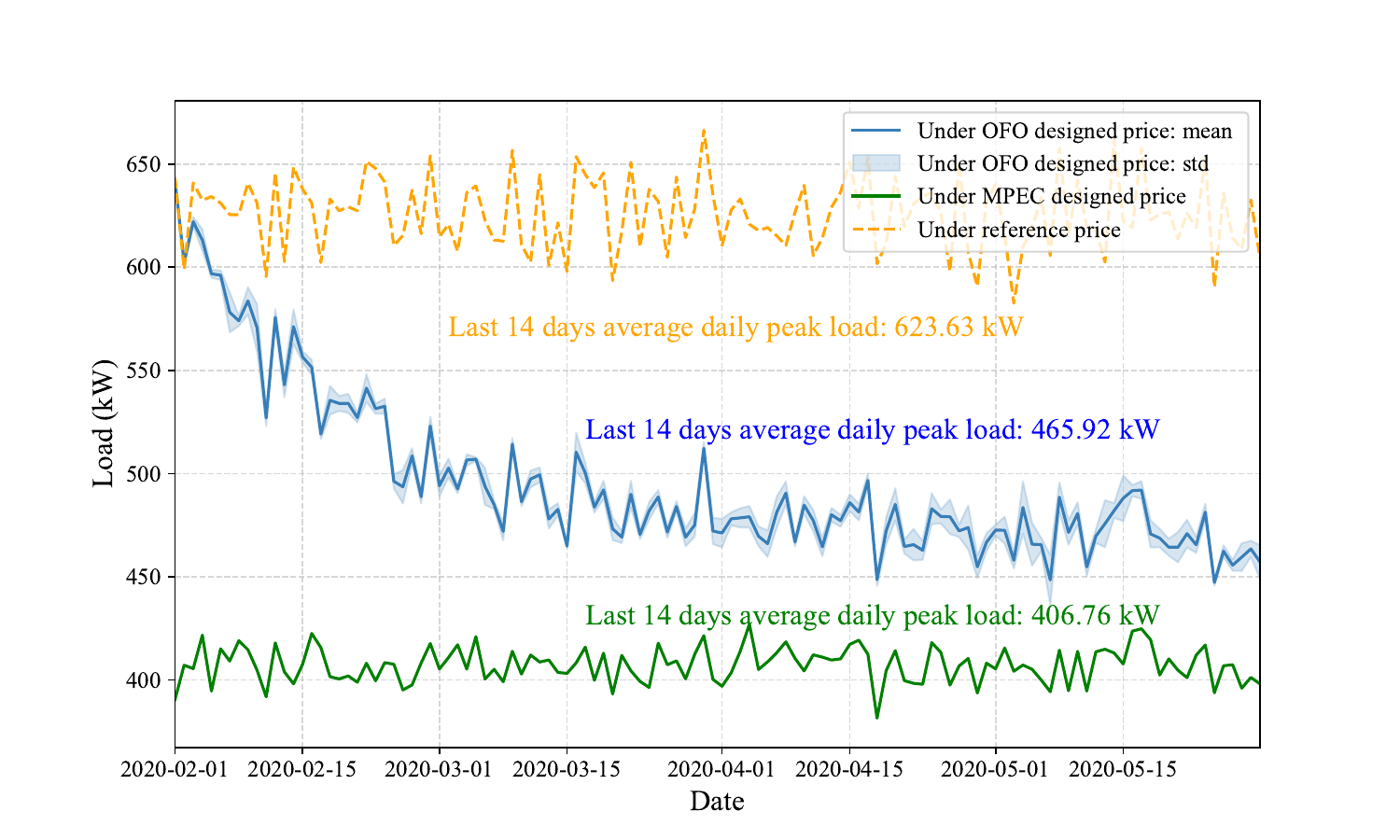}
    \caption{Comparison of Online Daily Peaks Loads of Three Methods in Four Months. 
    }
    \label{fig:re_dailyload}
\end{figure*}

\section{Results and Discussion}
\label{sec:results_discussion}

\subsection{Original and Smoothed OFO Comparison}
\label{sec:comp_origin_smooth}
\jr{
As noted at the beginning of Section~\ref{sec:settings}, the numerical evaluation uses the original max-based OFO implementation in Eq.~\eqref{eq: ofo_formulation}, while the smoothed formulation in Eq.~\eqref{eq: smooth_formulation} is mainly introduced for analytical purposes. Fig.~\ref{fig:comp_smooth} compares the aggregate load trajectories obtained from the original OFO and the smoothed OFO. \textbf{The two trajectories exhibit similar overall behavior, showing that the smoothed formulation preserves the main closed-loop response of the original method. However, the original max-based OFO achieves a slightly larger load reduction under the original problem setting.} In the smoothed formulation, the peak term is replaced by the unnormalized surrogate \(\log\sum_t \exp(\tau l_t)\), with \(\tau=10\). This relatively large value is needed because smaller values, such as \(\tau=1\), provide a weak peak-reduction signal in the closed-loop update. Since the unnormalized LSE approximates \(\tau\max_t l_t\), increasing \(\tau\) both sharpens the approximation and amplifies the peak-related gradient. The remaining difference between the two trajectories suggests that the LSE gradient does not fully reproduce the behavior of the exact maximum. Specifically, the LSE gradient distributes the marginal penalty across high-load intervals, whereas the max-function subgradient directly targets the active peak. Therefore, the smoothed formulation is useful for differentiability and convergence analysis, while the original max-based OFO remains preferable for the practical case-study evaluation.
}

\begin{table*}[t]
\centering
\caption{Comparison of last-14-days average daily peak load.}
\label{tab:peak_comparison}
\begin{threeparttable}
\begin{tabular}{lcccc}
\toprule
Method & Peak load (kW) & Reduction vs. baseline\tnote{1} & Gap to MPEC\tnote{2} & Gap / baseline peak\tnote{3} \\
\midrule
Reference price   & 623.63 & --        & --      & --   \\
MPEC benchmark    & 406.76 & 34.78\%   & --      & --   \\
OFO               & 465.92 & 25.29\%   & 14.6\%  & 9.5\% \\
\bottomrule
\end{tabular}
\begin{tablenotes}[flushleft,para]
\footnotesize
\item[1] Reduction vs. baseline:
$\frac{P_{\mathrm{Ref}}-P_{\mathrm{Method}}}{P_{\mathrm{Ref}}}\times 100\%.$
\item[2] Gap to MPEC:
$
\frac{P_{\mathrm{OFO}}-P_{\mathrm{MPEC}}}{P_{\mathrm{MPEC}}}\times 100\%.
$
\item[3] Gap / baseline peak:
$
\frac{P_{\mathrm{OFO}}-P_{\mathrm{MPEC}}}{P_{\mathrm{Ref}}}\times 100\%.
$

\end{tablenotes}
\end{threeparttable}
\end{table*}

\begin{figure*}[!h]
    \centering
    \includegraphics[width=0.7\linewidth]{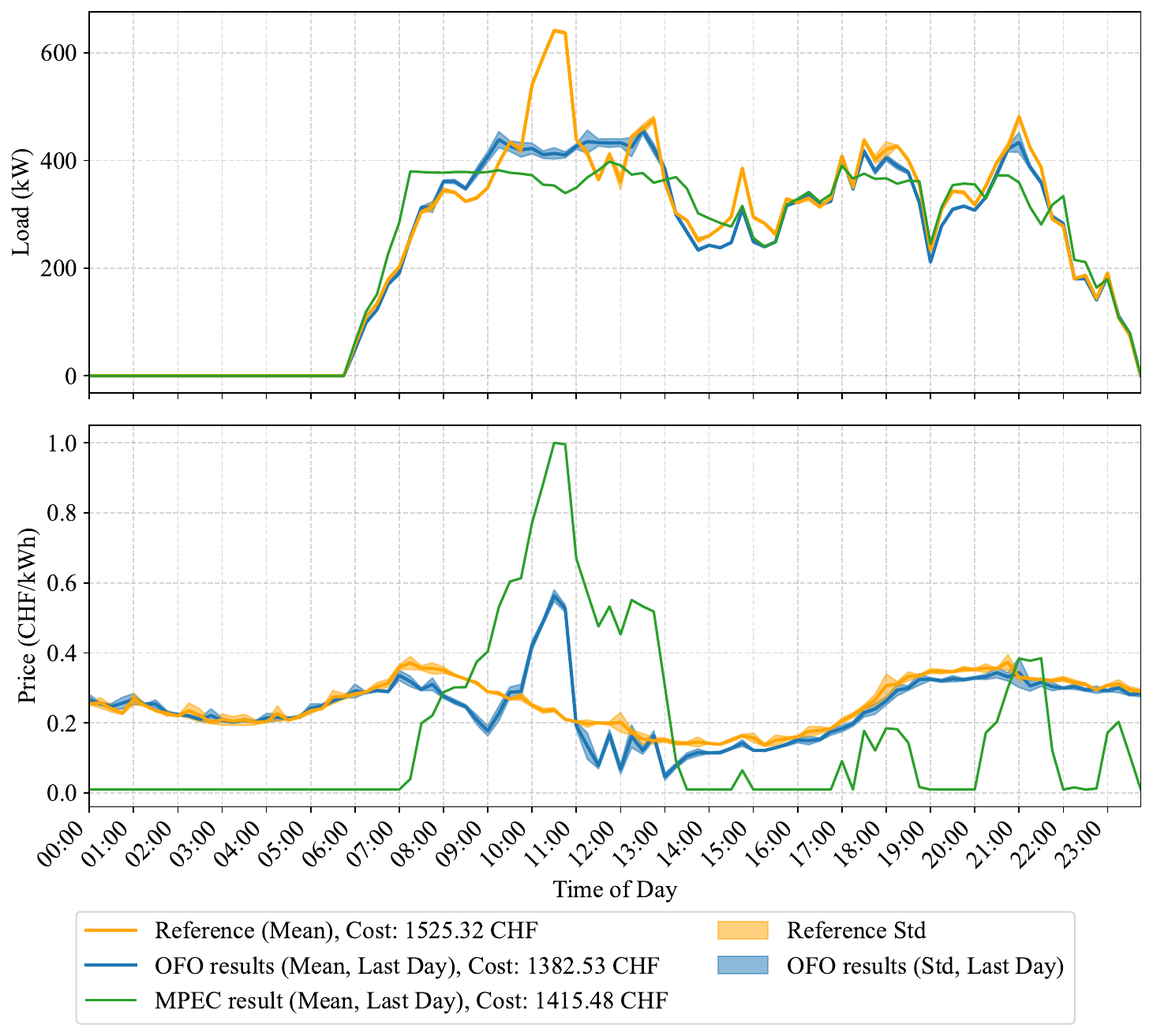}
    \caption{Comparison of daily load  and price profiles of the initial day and the last day. 
    }
    \label{fig:re_lastday}
\end{figure*}

\subsection{Peak-load Reduction Performance}

Fig.~\ref{fig:re_dailyload} compares the online daily peak load under (i) the proposed OFO-designed prices, (ii) the offline MPEC benchmark, and (iii) the static reference-price baseline. The MPEC benchmark is constructed from the static average dataset and therefore represents an offline full-information benchmark under a fixed daily charging pattern, rather than a predictive day-ahead solution for the actual next-day EV charging profile. In particular, since the next day’s charging sessions are not known in advance, the benchmark cannot pre-determine the truly optimal future price for the realized daily demand.

Table.~\ref{tab:peak_comparison} summarizes the peak-load results over the last 14 simulation days. The reference-price case gives the highest average daily peak, whereas OFO reduces this peak by $25.29\%$. The MPEC benchmark achieves the highest peak reduction, at $34.78\%$. Although OFO does not fully match the offline benchmark, the remaining difference is limited. Relative to the MPEC peak, the OFO peak is higher by only $14.6\%$. When expressed on the scale of the uncontrolled system, the same difference corresponds to just $9.5\%$ of the baseline peak. Together, these two measures indicate that OFO achieves performance close to the benchmark within a limited number of online updates, while still substantially reducing the original system peak.

Two observations are worth noting. First, OFO exhibits a transient adaptation phase early in the horizon, during which daily peaks decrease rapidly before stabilizing. This is expected in an online setting where the warm-start initialization provides a favorable starting point, and the price signal is progressively refined using sequential feedback. Second, although OFO follows the benchmark trend, a residual gap remains relative to MPEC. This gap is expected in finite-horizon operation because OFO adapts online using approximate sensitivities and only a limited number of updates. Nevertheless, our theoretical analysis shows that relying only on aggregate feedback does not preclude convergence under the stated assumptions. Consistent with this, longer-horizon simulations indicate that OFO continues to reduce peak load beyond the initial transient, although the rate of improvement diminishes over time. Overall, the observed stabilization and gradual peak reduction are consistent with the projected stochastic approximation result that the OFO iterates converge to the stationary set of the smoothened problem.

It is also important to emphasize the difference in information access and computational requirements between the two methods. The MPEC benchmark is solved offline with full access to accurate individual-level data and requires substantially higher computation; solving a 336-EV instance takes more than 2.6 hours on a personal computer. In contrast, OFO operates in closed loop using only aggregate measurements and lightweight updates, which makes it suitable for real-time deployment.

The shaded band in Fig.~\ref{fig:re_dailyload} shows the variability across the four initialization runs defined in Eqs.~\eqref{eq:init_price_noise_both}--\eqref{eq:init_price_noise_robustness}. The relatively narrow spread after the initial adaptation phase indicates that the proposed OFO method is robust to variations in the initial conditions and delivers consistent peak-reduction performance across runs.

\subsection{Characteristics of Learned Price Signals}
\label{sec:price_results}

Fig.~\ref{fig:re_lastday} shows the price and aggregate-load profiles on the last simulation day. Relative to the reference tariff, OFO raises prices during congested periods and lowers them during less congested hours, thereby inducing inter-temporal shifting of charging demand. This stronger temporal price contrast reduces the peak of $l_{\mathrm{agg}}$ and redistributes charging toward lower-price intervals while still satisfying the energy requirements. The MPEC benchmark follows the same qualitative mechanism but applies a more aggressive peak-hour price signal, resulting in a lower peak load on that day.

The cost values reported in Fig.~\ref{fig:re_lastday} further show that peak shaving is not achieved at the expense of higher total expenditure. On the last day, the reference case incurs a total cost of $1525.32$~CHF, whereas OFO reduces the cost to $1382.53$~CHF, corresponding to a decrease of $142.79$~CHF or $9.36\%$. The MPEC benchmark yields a cost of $1415.48$~CHF, which is $109.84$~CHF or $7.20\%$ lower than the reference case. Compared with MPEC, OFO attains a slightly lower cost on that day by $32.95$~CHF ($2.33\%$). This outcome is acceptable because it benefits both consumers and the DSO. Consumers pay less for charging under OFO than under the reference tariff, while the DSO benefits from lower peak-related operational stress and associated mitigation costs. It is also consistent with the objective design: peak reduction is prioritized, and the cost-deviation penalty is assigned a very small weight ($b=0.0002$). As a result, the penalty mainly serves as a safeguard against extreme deviations rather than enforcing strict cost matching. In practice, this trade-off can be tuned by adjusting $b$: a larger value places more emphasis on cost consistency, whereas a smaller value further prioritizes peak shaving.

\begin{figure*}[!t]
    \centering
    \includegraphics[width=0.7\linewidth]{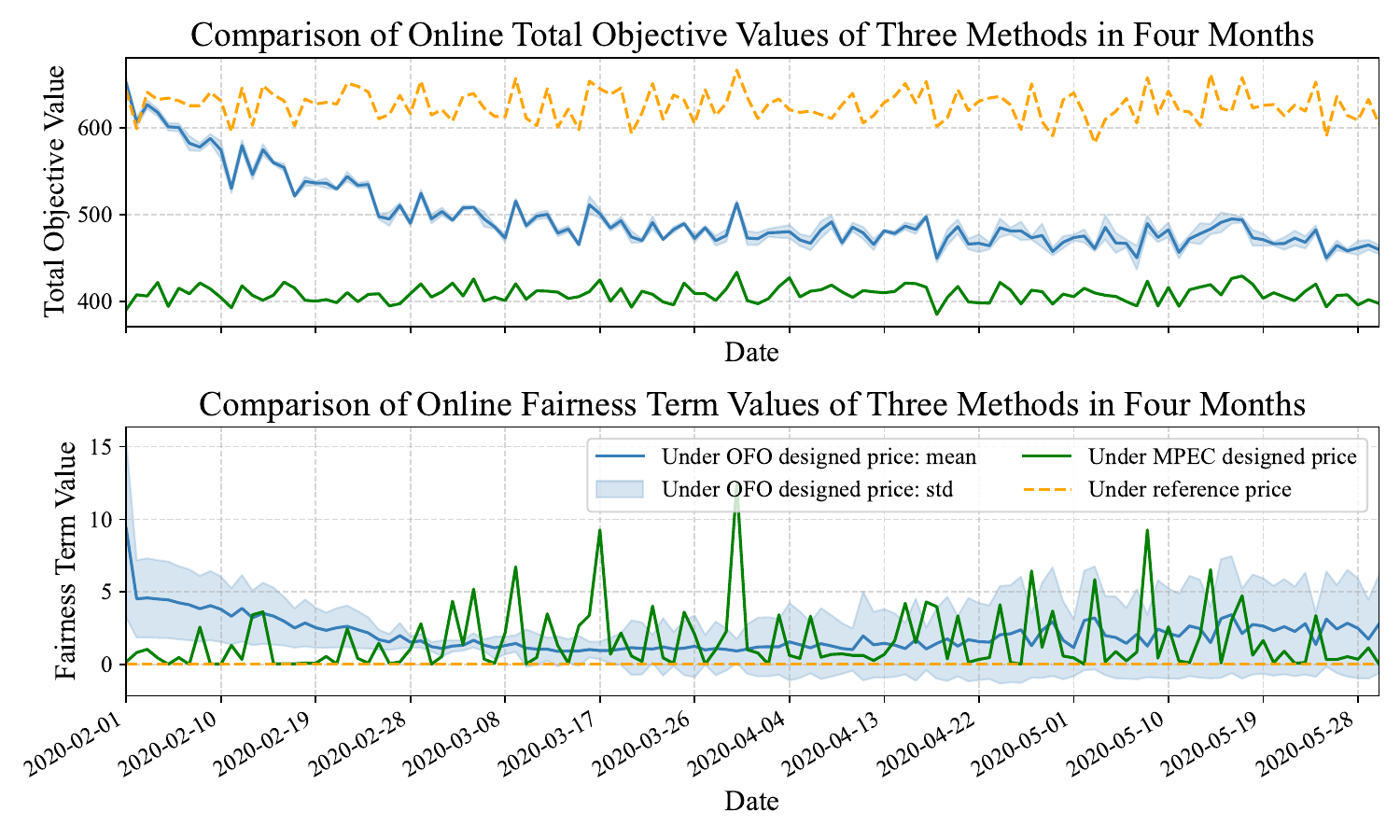}
    \caption{Comparison of Online Objectives of Three Methods in Four Months. 
    }
    \label{fig:re_objective}
\end{figure*}

Finally, the variability bands Fig.~\ref{fig:re_lastday} remain relatively narrow after the initial adaptation phase, indicating that the learned price profiles are robust to small perturbations in the initial prices. Together, these results show that OFO consistently learns a stable peak-shaving pricing pattern, characterized by stronger price separation between peak and off-peak hours while maintaining a broadly stable average price level.

\subsection{Objective Tracing}
\label{sec:objective_tracing}

Fig.~\ref{fig:re_objective} shows the online evolution of the DSO objective and its cost-deviation term for the three methods. Since the peak-load component has already been discussed in Fig.~\ref{fig:re_dailyload}, the focus here is on the overall objective value and the evolution of the cost deviation term.

In the top panel of Fig.~\ref{fig:re_objective}, the total objective under OFO decreases rapidly during the initial adaptation phase and then stabilizes, consistent with the progressive peak reduction observed in Fig.~\ref{fig:re_dailyload}. The MPEC benchmark achieves the lowest objective values throughout, as expected from an offline optimization baseline. Importantly, the OFO objective remains consistently below that of the reference-price case over the entire horizon, indicating that the proposed online updates effectively improve the DSO objective while maintaining the intended balance between peak reduction and cost deviation relative to a static tariff.

\begin{figure*}[!t]
    \centering
    \includegraphics[width=0.8\linewidth]{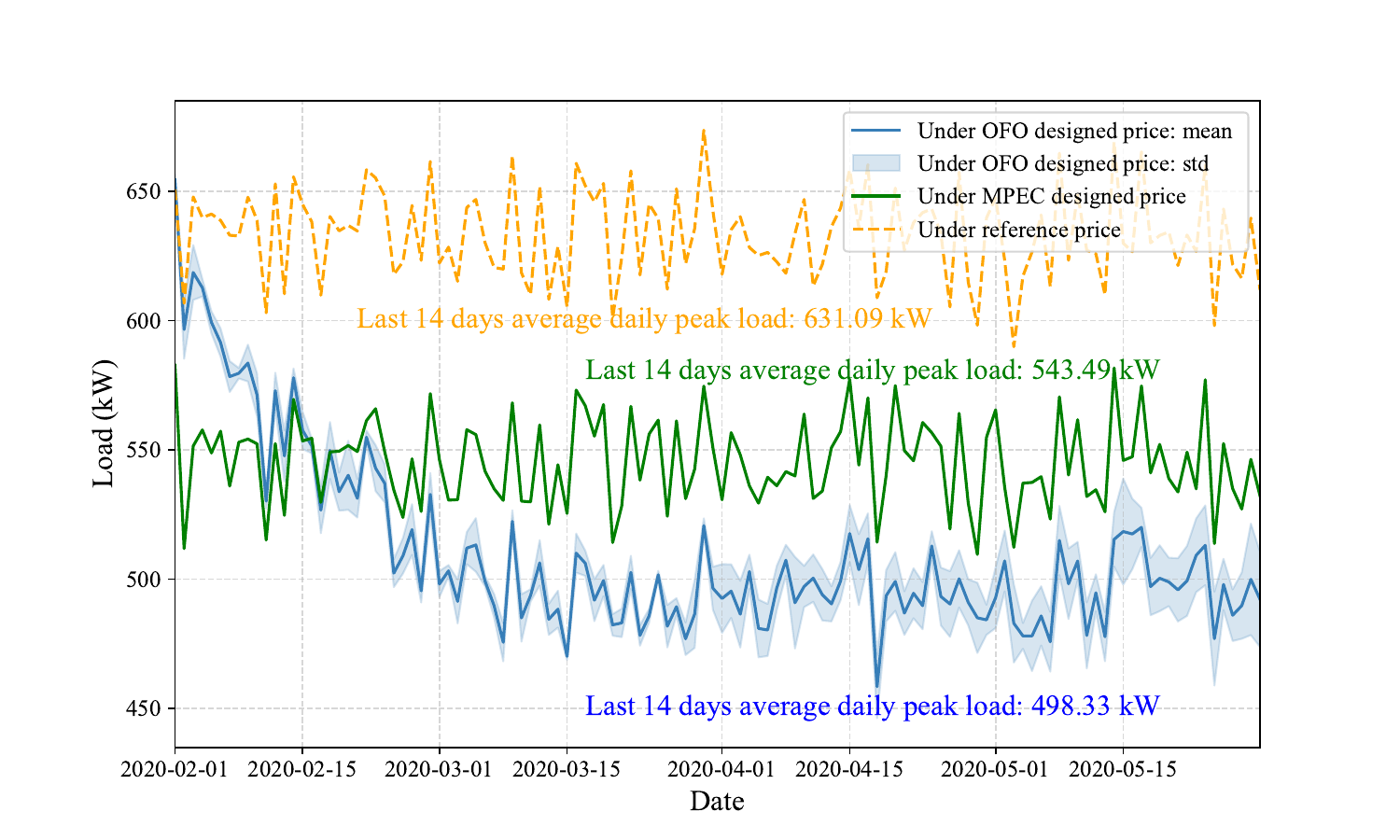}
    \caption{Comparison of Online Daily Peaks Loads of Three Methods in Four Months under Model Mismatch. 
    }
    \label{fig:model_mismatch}
\end{figure*}

The bottom panel of Fig.~\ref{fig:re_objective} shows the evolution of the cost deviation term over time. The reference-price case remains zero, as expected, because the cost-deviation term is defined relative to the reference price. Under OFO, this term remains bounded and generally small, with only transient increases during periods of stronger price adjustment. This behavior is mainly due to the OFO update mechanism. Since OFO is initialized close to the reference price and uses a small constant step size, the price trajectory evolves gradually, which helps keep the cost-deviation term small and avoids abrupt departures from the reference case. In contrast, the MPEC benchmark applies a fixed offline-optimized price trajectory. Because it is optimized for the general daily model, it may place stronger emphasis on peak reduction in certain periods, which can lead to larger cost deviations under other daily conditions and therefore occasional spikes in the deviation term.

\jr{
\subsection{OFO under Model Mismatch}
\label{sec:mismatch_result}
Fig.~\ref{fig:model_mismatch} evaluates the performance of OFO under the demand-timing mismatch introduced in the experiment settings (Section~\ref{sec:exp_setting}). In this case, the warm-start sensitivity matrix and the MPEC benchmark price are computed from the original charging-window model, whereas the online aggregate-load response is generated using charging windows delayed by one time step. Therefore, the initial sensitivity information available to OFO is structurally inaccurate.

Under this mismatch, the MPEC benchmark becomes suboptimal because its price is fixed after being computed from the nominal lower-level model. The resulting last-14-day average daily peak load is $543.49$~kW, indicating that the MPEC solution does not fully match the true delayed-window system. In contrast, OFO updates its sensitivity estimate online using the observed aggregate-load response. This adaptation reduces the last-14-day average daily peak load to $498.33$~kW, which is $45.16$~kW lower than the MPEC benchmark.
These results show that OFO can recover useful price-response information even when the warm-start sensitivity matrix is computed from an inaccurate lower-level model. Compared with the MPEC benchmark, OFO is more robust to model mismatch because it uses aggregate feedback to update the sensitivity estimate online instead of relying solely on a fixed nominal lower-level model. However, compared with the OFO performance in Fig.~\ref{fig:re_dailyload}, the mismatched warm start leads to a smaller peak reduction over the same operating period. \textbf{Overall, model mismatch degrades both methods, but its effect is more severe for the fixed MPEC benchmark, whereas OFO can partially compensate for the mismatch through online feedback-based adaptation.}

To be noted, the OFO performance under model mismatch is nevertheless weaker than in the non-mismatched case. This is expected because the delayed charging windows change the feasible response structure of the lower-level agents. As a result, the warm-start sensitivity matrix does not only contain magnitude errors, but also assigns price-response effects to inaccurate time periods. During the early online iterations, OFO therefore updates the price using an imperfect sensitivity direction, which slows convergence and may lead to less effective peak shifting. 
}

\subsection{Practical Considerations: Scalability, Data Access, and Constraint Violations}
\label{sec:practical_considerations}

Compared with offline full-information benchmarks, OFO generally requires more \emph{real days} to converge to an effective price profile that reduces peak load while preserving cost similarity, because the prices are refined online using sequential aggregate feedback. This slower adaptation is the natural trade-off of operating under limited information with lightweight updates and without access to accurate individual-level EV data. In practice, this trade-off is justified because the next day’s EV charging profile is not known in advance, so a truly optimal day-ahead price cannot be computed exactly for the realized demand.

For this reason, OFO is the more practical choice when data access is restricted and the EV population is large. The method scales naturally with the number of EVs because the DSO update relies only on the aggregate load signal, and it avoids the substantial computational burden of solving offline bi-level or game-theoretic formulations at large scale. In particular, solving an MPEC repeatedly in daily operation is impractical, both because accurate next-day EV profiles are unavailable and because the computation time is high. By contrast, OFO can be updated directly from measured system response and is therefore better suited to real-time or repeated operational use. From an operational perspective, OFO also exhibits stable behavior during the adaptation process. In the simulations, the iterative updates rarely violate the imposed operating limits, which is consistent with the use of bounded prices, explicit peak-load constraints, and a conservative constant step size to suppress large transients.

\section{Conclusion}
\label{sec:conclusion}

This paper investigated online feedback optimization for dynamic electricity pricing in heterogeneous distribution systems, using EV charging as a representative flexible demand. A smooth reformulation was introduced to facilitate convergence analysis, and the projected OFO updates are shown to converge to the projected stationary set of the problem in a limited-information setting, where the DSO observes only aggregate load measurements.

Simulation results over four months show that OFO achieves substantial and consistent peak reduction. Over the last 14 simulation days, the average daily peak load is reduced by $25.3\%$ relative to the reference tariff. Compared to the minimum peak load attained in an offline full-information benchmark, the proposed method incurs a gap of $14.6\%$. This result is particularly relevant from a practical perspective, since OFO operates with aggregate measurements only and lightweight online updates, whereas the benchmark requires accurate individual data and substantially higher computational effort. \jr{Additionally, robustness tests under multiple initializations show limited variability after the initial adaptation phase, while the model-mismatch test demonstrates that OFO can maintain effective performance when the lower-level model is inaccurate.}

Overall, the above results indicate that OFO provides an effective and practical solution for peak-demand management under limited observability. The proposed method features efficient peak load reduction and moderate computational requirements, making it preferable in real-world applications where detailed user-level data are unavailable.

Future work will extend the proposed method to settings with periodic non-stationarity, such as weekly demand patterns with weekend effects. Other important directions include adaptive tuning of the peak--cost trade-off, incorporation of additional network constraints and uncertainty sources, and assessment under partial participation in large-scale deployments. It is also important to relax the assumption of fully rational user response by considering heterogeneous behavioral patterns, including partial, delayed, or absent responses to price signals.

\section*{Declaration of Generative AI and AI-assisted Technologies in the Manuscript Preparation Process}

During the preparation of this work the author(s) used ChatGPT in order to improve language clarity and conciseness. After using this tool, the author(s) reviewed and edited the content as needed and take(s) full responsibility for the content of the publication.

\section*{Acknowledgement}

This work was supported by NCCR Automation, a National Centre of Competence in Research, funded by the Swiss National Science Foundation (grant number 51NF40\_225155) and Max Planck ETH Center for Learning Systems.

\section*{CRediT Authorship Contribution Statement}
\noindent\textbf{Jiarui Yu:} Conceptualization, Methodology, Investigation, Writing - Original Draft. \textbf{Zhiyu He:} Conceptualization, Supervision, Writing - Review \& Editing. \textbf{Wenbing Wang:} Supervision, Writing - Review \& Editing. \textbf{Colin N. Jones:} Writing - Review \& Editing, Funding acquisition. \textbf{Florian Dörfler:} Writing - Review \& Editing, Funding acquisition. \textbf{Hanmin Cai:} Conceptualization, Supervision, Resources, Writing - Review \& Editing, Funding acquisition.

 \bibliographystyle{elsarticle-num-names} 
 \bibliography{cas-refs}





\end{document}